\documentclass[letterpaper]{article}
\usepackage[margin=1in]{geometry}
\usepackage{graphicx}
\usepackage{natbib} 
\usepackage{bbm}
\usepackage{amsthm}                
\usepackage{amssymb} 
\usepackage{amsmath} 
\usepackage{thmtools,thm-restate}
\usepackage[autostyle=true,german=quotes]{csquotes}
\PassOptionsToPackage{hyphens}{url}\usepackage{hyperref}
\usepackage{cleveref}

\newtheorem*{theorem*}{Theorem}
\newtheorem{lemma}{Lemma}
\newtheorem*{lemma*}{Lemma}

\newcommand{\p}[1]{\left( #1 \right)}

\newcommand{\abs}[1]{\left \vert #1 \right \vert}

\newcommand{\cd}[0]{\cdot}

\newcommand{\val}[0]{\ensuremath{v}}

\newcommand{\perm}[0]{\ensuremath{\sigma}}
\newcommand{\permdist}[0]{\ensuremath{\mathcal{P}}}
\newcommand{\nitem}[0]{\ensuremath{N}}
\newcommand{\pfree}[0]{\ensuremath{p}}
\newcommand{\avail}[0]{\ensuremath{s}}
\newcommand{\busypen}[0]{\ensuremath{\gamma}}
\newcommand{\ratio}[0]{\ensuremath{r}}

\newtheorem{observation}{Observation}
\newtheorem{definition}{Definition}
\usepackage[ruled, linesnumbered]{algorithm2e}
\SetKwComment{Comment}{/* }{ */}

\newcount\Comments
\usepackage{color}
\usepackage{xcolor}
\definecolor{darkgreen}{rgb}{0,0.5,0}
\newcommand{\kibitz}[2]{\ifnum\Comments=1{\color{#1}{#2}}\fi}

\newcommand{\updated}[1]{{\color{purple}#1}}

\newif\ifinformal
\informalfalse 

\begin{document} 

\title{Optimal Selection Using Algorithmic Rankings with Side Information}

\author{%
  Kate Donahue\thanks{Massachusetts Institute of Technology}\thanks{Work partially completed while at Microsoft Research and Cornell University.} 
  \\
   Nicole Immorlica\thanks{
   Yale University \& Microsoft Research} \\
  Brendan Lucier \thanks{
  Microsoft Research} \\
}
\date{} 

\maketitle
\begin{abstract}
Motivated by online platforms such as job markets, we study an agent choosing from a list of candidates, each with a hidden quality that determines match value.  The agent observes only a noisy ranking of the candidates plus a binary signal that indicates whether each candidate is \enquote{free} or \enquote{busy}.   Being busy is positively correlated with higher quality, but can also reduce value due to decreased availability. We study the agent's optimal selection problem in the presence of ranking noise and free-busy signals and ask how the accuracy of the ranking tool impacts outcomes. In a setting with one high-valued candidate and an arbitrary number of low-valued candidates, we show that increased accuracy of the ranking tool can result in suboptimal social outcomes. For example, increased accuracy may mean that agents may be more likely to make offers to busy candidates, and (counter-intuitively) may be more likely to select lower-ranked candidates.  We further discuss conditions under which these results extend to more general settings.
\end{abstract}

\section{Introduction}

Consider the following setting: you are a hiring manager, trying to find a candidate to interview for a hiring slot. You know that some candidates are better suited for your role than others - ideally, you would like to interview the best candidate available. To assist you with this goal, you might obtain a (noisy) ranking of the candidates from a commercial employment site, such as LinkedIn, Indeed, or other ranking tools. If the ranking is \enquote{good} (if the expected utility of candidates decreases as you navigate further down the list), the best action for the employer is clear: always pick the top-ranked candidate. Also, as the ranking tool becomes more accurate, the employer is only happier with their choice, as the expected value of the top-ranked candidate only increases. 

However, in many cases, employers have access to side information, besides the ranking. For example, a hiring manager may know that some candidates are less likely to accept any eventual offer because they are already employed (see \cite{filippas2023advertising} for an empirical example of such a setting).  
We call such a candidate \enquote{busy}, as compared to a \enquote{free} candidate without these competing sources of employment. If a busy candidate would \emph{never} accept an offer, then the choice is clear: the best option is to select the highest-ranked free candidate. However, if a busy candidate still has some chance of accepting an offer, the possible strategy space changes. In many markets there is an opposing force to the penalty for being busy: it may be the case that high-value candidates are more likely to be busy, and thus being busy is a signal of candidate quality. For example, this could be because high value candidates are more likely to have already been identified and hired by competing firms\footnote{Throughout this paper, we will make the assumption that some candidates are better fits for the job than others (e.g. higher or lower value). We recognize that this assumption drastically simplifies the multi-faceted strengths and weaknesses each candidate may bring, and discuss this (and other) assumptions in more depth in Section \ref{sec:modelassump}.}. In many settings it may not be possible to resolve the uncertainty inherent in a busy signal: for example, the candidate themselves may not know whether or not they would be willing to move without going through the interview process, at which point the firm has already used up its interview slot. 

Here, an employer faces a difficult question: given a noisy ranking of candidates, along with information about which ones are free or busy, which candidate should she select? Should she pick the top-ranked busy candidate, following the old adage \enquote{if you need something done fast, ask a busy person}? Should she aim to avoid the extra hassle of competing for a busy candidate and instead hire the top-ranked free candidate? Moreover, how should the firm's strategy change if the ranking tool becomes more accurate? 

Moreover, a firm's hiring strategy has impacts on the rest of society. For example, consider that a firm's choice to interview a busy (already employed) candidate over a free (unemployed) candidate may be more or less valuable to the candidate. An unemployed candidate may have higher utility for an interview than an employed candidate. Alternatively, an employed candidate may have a stronger ability to negotiate a better offer, which could increase their welfare more. While we do not directly model such externalities, we will use \enquote{probability of picking a busy candidate} and other such quantities as benchmarks to understand the impact of a ranking tool on broader society. Specifically, we will be interested in how these quantities change as the accuracy of the ranking tool increases. Do these measures uniformly improve or degrade as the firm uses a more accurate ranking tool?

\noindent\textbf{Motivation:} More broadly, in this paper our goal is to study a stylized model of an agent using an algorithmic tool to make a decision, and combining the information encoded by the tool with her own source of information. Given the importance attached to continually increasing the accuracy of algorithmic tools, we wish to show that when these tools are used as part of a broader system involving human discretion and alternative sources of information, the returns to increased algorithmic accuracy are less clear: in particular, relevant measures of welfare could \emph{deteriorate} under increased accuracy. We hope that our results will lead to further research in other scenarios where algorithmic tools are used alongside other sources of information, and in particular into what properties of such algorithmic tools are helpful for societal welfare. While we will use the employer hiring job candidates as our motivating example, we note that this setting is quite broad and could model a range of other scenarios where humans make decisions. For example, consider a hungry customer navigating a ranking of restaurants, each \enquote{crowded} or \enquote{not crowded}: a crowded restaurant is more likely to be good, but also has a longer wait. 

\noindent\textbf{Summary:} The rest of this paper is as follows: in Section \ref{sec:modelassump}, we will present our model and assumptions. In Section \ref{sec:related}, we discuss the connection of our paper to related work, including the extensive literature on herding and the Pandora's Box problem. Next, Section \ref{sec:optstrat} analyzes the optimal strategy that a firm navigating a ranking might take: should she pick the first busy candidate, the first free candidate, or some other candidate entirely? In Section \ref{sec:welfare} we analyze phenomena related to welfare, like how frequently firms preferentially select job candidates who are already employed (potentially modeling unemployment discrimination) or avoid doing so (potentially modeling collusion in job hiring). We also analyze how often firms select the top ranked candidate, a measure of how incentive compatible the ranking is with potential implications for the welfare of the company creating the ranking tool. Throughout, we consider how increasing the accuracy of the ranking tool can affect the strategy, welfare, and fairness results in previous sections. Surprisingly, we show that increased accuracy can cause \emph{negative} outcomes. In Section \ref{sec:beyondsuperstar} we discuss extensions of our model and in Section \ref{sec:discuss} we conclude by discussing implications of our results. All proofs are deferred to the appendix.

\section{Model, notation, and assumptions}\label{sec:modelassump} 

 \textbf{Model and notation: }
First, we present our model. There are $\nitem$ candidates, of which the agent (alternatively, firm) wishes to select exactly one: for example, to interview a job candidate. The candidates have different true values to the firm, where $\val_i$ denotes the value of candidate $i$, though the firm \emph{cannot} know the true value of candidates at each ranking (this only holds in the limit of a perfectly noiseless algorithmic ranking tool). The goal of the firm is to maximize their expected utility. In order to help with this decision, the firm has access to a ranking tool which produces noisy permutations over the candidates: $\perm \sim \permdist$, where $\perm_i$ denotes the index of the candidate ranked $i$th in permutation $\perm$ with value $\val_{\sigma_i}$. We will require that the expected value of candidates decreases in the ranking ($\mathbb{E}_{\perm \sim \permdist}[\val_{\perm_i}]$ is decreasing in $i$). 

\noindent \textbf{Permutation distribution:} When comparing two distributions $\permdist$ and $\permdist'$, we will use notation $P[\sigma]$ and $P[\sigma]'$ to denote the probability $\permdist$ and $\permdist'$ place on $\sigma$, respectively. In general, we will assume that the probability of observing each permutation is drawn from a Plackett-Luce distribution \cite{plackett1975analysis, luce1959individual}). This commonly-used model is where i.i.d. noise is added to the true value of each item to produce $\hat \val_i = \val_i + \epsilon$, for $\epsilon$ drawn from a Gumbel distribution: then, the items are sorted in order of the noised values $\{\hat \val_i\}$. Definition \ref{def:accuracy} defines how we will consider changing accuracy throughout most of the paper: 
\begin{definition}[Increased accuracy]\label{def:accurate}
A permutation distribution governed by Plackett-Luce has \emph{increased accuracy} when it has a smaller degree of noise (smaller standard deviation) in the Gumbel distribution.
\end{definition}
However, many of our results rely only weakly on the Plackett-Luce model, and in the relevant appendices we state the sufficient properties on the permutation distribution in order for our results to apply. 

\noindent \textbf{Firm strategy:} In addition to the ranking, the firm has access to a single bit of side information about each candidate: specifically, a status vector $\avail$, where $\avail_i = 1$ if the candidate ranked in position $i$ is free, and $\avail_i = 0$ if the candidate is busy. We assume that higher-valued candidates are less likely to be available, in particular, that $\val_i \geq \val_j$ implies $\pfree_i \leq \pfree_j$, where $\pfree_i$ gives the probability that candidate $i$ is free. We also assume that each candidate's probability of being free or busy is independent of every other candidate: specifically, each free/busy realization is drawn from a Bernoulli distribution $s_i \sim Bern(p_i)$ before the firm observes the ranking. Throughout this paper we will take these probabilities as fixed and assume each candidate reports her free-busy status honestly: in the conclusion we will discuss extensions of this model. 

To summarize, there are three steps that occur in this model: 
\begin{enumerate}
    \item A permutation over candidates is drawn, with some $\perm$ drawn from a Plackett-Luce distribution based on values of the candidates. 
    \item Each candidate obtains a free/busy status, independently given her probability $p_i$ of being free. 
    \item A firm may pick a single candidate, observing \emph{only} the free-busy status and not the candidate values. 
\end{enumerate}

Picking a busy candidate incurs a penalty: specifically, if a candidate has value $\val_i$ if picked when they are free, then they have value $\val_i/\busypen$ if they are selected when they are busy, for $\busypen \geq 1$. In general we allow $\busypen_i$ to vary with the candidate $i$: when we use $\busypen$ without a subscript, we will consider the simplifying case when $\busypen_i$ is identical across candidates. This $\busypen$ parameter reflects the fact that a busy (e.g. already employed) job candidate is less likely to accept an offer. We may also view $\busypen$ as a firm-specific parameter reflecting, for example, the attractiveness of that firm to candidates\footnote{We do not directly model candidate utilities, but include this solely as interpretation for when $\busypen$ may be high or low.}.  

One central assumption is that the firm must select a candidate based solely on the ranking, and cannot re-select another candidate later. This could be motivated by hiring in stages, where \enquote{selecting a candidate} corresponds to bringing a candidate on for an intensive onsite which cannot easily be filled with another candidate if one declines an offer after the onsite. Thus, we model strategy involved in this first stage of hiring, e.g., the selection of a candidate for an onsite interview. We also assume that it is not possible to resolve the uncertainty inherent in the free/busy signal - for example, while the employer could ask the candidate about their willingness to move jobs, the candidate themselves may not know perfectly without going through the recruiting process.

\noindent \textbf{Superstar model:} In much of this paper,  we will focus on the \emph{super-star setting}, with exactly one high-value candidate, and all other candidates with lower value: $\val_1 \geq \val_2 = \val_3 \ldots  \geq 0$ and $\pfree_1 \leq  \pfree_2 = \pfree_3 \ldots$. This models settings where the distribution of candidate quality is skewed: the highest-value individual may have much higher value than other candidates, who are all roughly comparable. In this setting, the space of permutations becomes much smaller, because the only distinguishing feature is the index of the high-value item. We will use the notation $\perm^i$ to denote the permutation where the high-value item is in the $i$th index, and in general we will assume that $P[\perm^i] \geq P[\perm^{i+1}]$ for all $i$ (that is, the ranking is more likely to rank the high value item higher). In Section \ref{sec:beyondsuperstar} we will show that while the general setting is more complex, our core results generalize beyond the superstar model. 

\section{Related work}\label{sec:related}
Our model has connections to celebrated models of agents strategically picking items, such as the hiring problem and Pandora's box problem \cite{krengel1977semiamarts, weitzman1978optimal}. However, our setting has key differences from these. Foremost among them is the signaling mechanism: rather than each item having an independent, known distribution of value, we model the scenario where the free/busy signal is generated by items with different probabilities, given different values. This structure, in conjunction with the ranking, means that the free/busy status of item $i$ can affect the expected value of item $j$ -- generally not captured in existing models. Furthermore, in our setting the agent must pick a single item and commit to it without being able to observe its true value. This captures the setting where it is difficult or impossible to evaluate the item without committing to it first -- for example, in hiring, when a candidate's true \enquote{value} may not become apparent until after weeks or months of work.

Additionally, our work has connections to long literature on herding (also known as information cascades) \cite{banerjee1992simple, bikhchandani1992theory, welch1992sequential}, a pattern where it may be optimal for a decision-maker to follow actions taken by previous people, even ignoring their own information. Such patterns can lead to \enquote{cascades} where multiple sequential decision-makers each follow actions taken by previous decision-makers, rather than following their own information. In our setting, when an agent chooses to pick a lower ranked busy candidate over a higher ranked free, this could be seen as herding, since the free/busy status of candidates could be seen as created by actions taken by previous decision-makers, whereas the ranking is assumed to be private (personal) information. To our knowledge, herding has not been extensively studied in ranking settings: our model interestingly allows for herding to be moderated by the strength of private information through the quality of the ranking and the index of the first busy candidate. 

Our paper also has connections to the literature on algorithmic monoculture \cite{kleinberg2021algorithmic}, which studies when utility may be \emph{decreased} when two agents rely on the same algorithmic ranking, rather than their own (uncorrelated) rankings. Variants of this setting have been explored, such as how algorithmic monoculture can lead to individuals experiencing correlated outcomes (\enquote{outcome homogenization}) \cite{NEURIPS2022_17a234c9}, or the implications of monoculture and of noisy matching in two-sided matching markets. For example, \cite{peng2023monoculture, peng2024wisdom, baek2025hiring} all model settings where firms rely on common algorithmic tools to make hiring decisions, and the potential pitfalls of such an approach.

Common features include the model of multiple agents taking strategic actions in relation to imperfect rankings, as well as the notion of a penalty for picking an item that another agent has picked. Key differences in our model include the fact that this penalty is softened -- in particular, an agent can derive non-zero utility for picking an item that another has already chosen. This dramatically complicates the strategy space. Multiple works have considered the question of fairness in ranking - \cite{zehlike2021fairness} offers a helpful survey. In particular, \cite{NEURIPS2021_63c3ddcc} describes a notion of fairness relating to how often items are presented in the top $k$, given that they have some probability of truly being the best $k$ items (relating to our notion of an item being \enquote{picked}). Issues of fairness and strategy have also been studied in the related area of algorithmic monoculture \cite{blackraghavanbarocas22, jain2023algorithmic, cooper2023my, NEURIPS2022_ba4caa85}, while \cite{peng2023reconciling} studies the tradeoff between diversity and accuracy, showing that consumption constraints (e.g. only top item can be consumed) explains away a tension between these two goals. 

Finally, there are multiple papers (empirical and theoretical) exploring models of firms using prediction tools to compete with each other. For example, \cite{filippas2023advertising} explicitly studies an empirical version of this where job candidates could pay to send a signal of availability to potential employers. Intriguingly, in this setting it seemed like such a signal did \emph{not} lead employers to preferentially pick busy candidates, with most employers choosing a \enquote{first-free} strategy, perhaps indicating a parameter regime where this strategy was typically dominant. \cite{Michellespaper} explores this setting in a dynamical pricing model, giving conditions where pricing levels for such a signal induce a truthful response. Additionally, \cite{jagadeesan2023improved} explores a setting where improved data representation (as modeled by reduced Bayes risk) can paradoxically increase the error that users experience when they choose between two competing firms. \cite{10.1145/3490486.3538364} studies statistical discrimination in stable matchings where candidates have preferences over firms, who only have noisy access to signals about candidate quality, and \cite{castera2024correlation} studies a setting where different candidates have differing degrees of correlation in their predictions across models, which leads to tensions between societal welfare and individual candidate welfare. 

\section{Superstar: Optimal Selection Strategies}\label{sec:optstrat}

In this section, we begin by evaluating the firm's strategy: given access to an algorithmically-generated ranking over candidates and each candidate's free-busy status, what is the optimal strategy? (All proofs will be deferred to Appendix \ref{app:optstrat}.) First, we will show that (given a status vector $\avail$) there are only two strategies we need to consider within the superstar setting: picking the first free candidate, or the first busy candidate. This dramatically reduces the strategy space firms need to consider. 

\begin{restatable}{lemma}{alwaysborf}
\label{lem:alwaysborf}
In the superstar setting, given a realized status vector $\avail$, it is always optimal to either pick the first (top-ranked) free item or the first busy item. 
\end{restatable}

Next, the question turns to \emph{which} of these two candidates (top-ranked free or top-ranked busy) we should pick, and the more general strategy the firm follows. Algorithm \ref{algo:superstar} describes such a strategy. In particular, this strategy is made up of two decisions: 1) Whether the firm seeks to preferentially select free or busy candidates, and 2) How far down the ranking the firm is willing to look for a candidate of their preferred free/busy status (their \enquote{window size}). Specifically, it shows that the strategy space depends on two terms. First, free-busy ratio $\ratio = \frac{p_2/(1-p_2)}{p_1/(1-p_1)}$ measures how strongly status is correlated with value - effectively, it shows how strongly observing whether a candidate is free or busy should influence our belief about the (unobserved) candidate value. Secondly, $\busypen$ measures the penalty for selecting a busy candidate - if this is larger, then the firm is \emph{less} likely to pick a busy candidate. 
The term $R = \frac{v_1/\busypen_1 - v_2}{v_1 - v_2/ \busypen_2} \cd r$ incorporates both $\ratio$ and $\busypen_1, \busypen_2$. If $R \leq 1$, firms have a preference for free candidates, and if $R\geq 1$, firms have a preference for busy candidates. However, the further down the ranking the firms search, the less likely it is that a given candidate will be high value. Therefore, the optimal strategy is for firms to pick the top-ranked candidate of their preferred status, so long as that candidate doesn't come \enquote{too far} down the ranking, where \enquote{too far} is measured in terms of the accuracy of the ranking tool. 

\updated{
\begin{algorithm}[htbp]
\caption{Given candidate values $\val_1> \val_2$, probabilities of being free $\pfree_1 < \pfree_2$, accuracy of ranking tool $\{P[\perm^i]\}$, $s \in \{0, 1\}$ for status of top-ranked candidate, and index $j$ of first candidate with different free-busy status from $s$. }
\label{algo:superstar}
\DontPrintSemicolon
\LinesNumbered
Define $R = \frac{v_1/\busypen_1 - v_2}{v_1 - v_2/ \busypen_2} \cd r$. \\
\eIf{$R\leq 1$}{
Calculate $j^* = \max_{j \in N}$ such that $ P[\perm^1]/P[\perm^j] \leq 1/R$, and set $j^*=1$ if no such $j$ exists. \\
Calculate $k = \min_{i \in [j^*]}$ such that $ s_k = 1$, and $k=1$ if no such index exists.  
}{
Calculate $j^* = \max_{j \in N}$ such that $ P[\perm^1]/P[\perm^j] \leq R$, and set $j^*=1$ if no such $j$ exists. \\
Calculate $k = \min_{i \in [j^*]}$ such that $ s_k = 0$, and $k=1$ if no such index exists. 
}
\Return{$k$} 
\end{algorithm}

}

\begin{figure}
    \centering 
\includegraphics[width=3.3in]{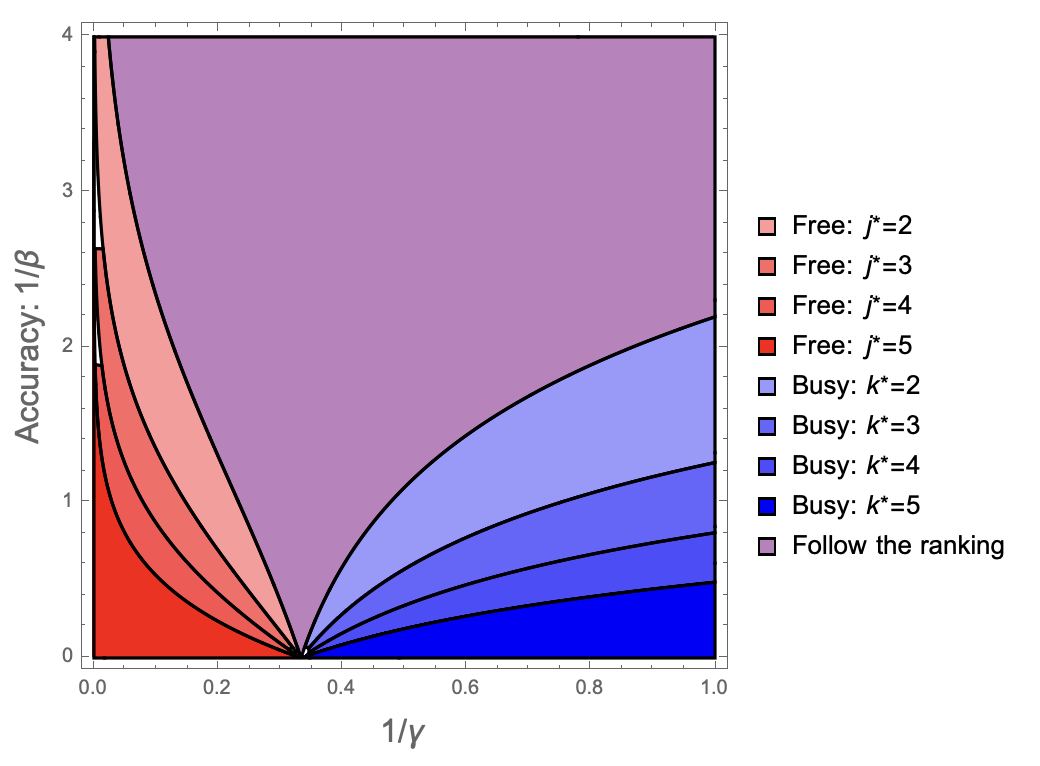}
\caption{Illustration of strategy in Algorithm \ref{algo:superstar} for firm, given $\nitem=5$ candidates with $\val_1 > \val_2=0$ and permutations given by Plackett-Luce. The x axis varies the $\busypen$ busy penalty, while the y axis varies accuracy as parameterized by the $1/\beta$ Gumbel noise parameter (higher values increase accuracy). Shades of {\color{red} red} indicate regions where first-free selection is the optimal strategy and shades of {\color{blue} blue} indicate where first-busy selection is optimal. Darker shades indicate regions where the firm has larger $j^*$ or $k^*$ and thus \enquote{hunts} further down the list to find a candidate with their preferred free/busy status. The {\color{purple} purple} region is where the optimal strategy is \enquote{follow the ranking} (pick the top ranked candidate always). Note that the choice of first-free or first-busy depends only on $\busypen$, while the window size depends on the accuracy of the ranking tool.
}

\label{fig:optstrat}
\end{figure}
For simplicity, we will refer to the three different strategy families that Algorithm \ref{algo:superstar} implies as: \textbf{$k$-free} if the firm is preferentially picking a free candidate within a \enquote{window} of size $k$, \textbf{$k$-busy} if the firm is preferentially picking a busy candidate within a \enquote{window} of size $k$, and \textbf{follow the ranking} if the window is of size 1, implying that the firm always ignores its own free/busy information and simply follows the ranking tool. Figure \ref{fig:optstrat} includes a visual illustration of this strategy space as a function both of the busy penalty $\busypen$ and the accuracy of the ranking tool. 

First, we will note one important property of the algorithm in Algorithm \ref{algo:superstar}: as accuracy increases (e.g. noise in the RUM decreases), the optimal window size of the algorithm only shrinks. This should make intuitive sense: firms know that the high value candidate is more likely to be at the top of the list, and thus search less far for candidates with their preferred free/busy status.

\begin{lemma}\label{obs:vpos}
As accuracy increases, the window size $j^*$ in Algorithm \ref{algo:superstar} stays constant or shrinks.
\end{lemma}

Finally, we will analyze the performance of the Algorithm in \ref{algo:superstar}. As a warm-up, we will show that the there exist cases where the \enquote{optimal} strategy differs from that in Algorithm \ref{algo:superstar}. 

\begin{restatable}{lemma}{counter}\label{lem:counter}
In the limit of an extremely uninformative ranking, there exist scenarios where the optimal strategy differs from Algorithm \ref{algo:superstar}. 
\end{restatable}
\begin{proof}[Proof sketch]
Consider $\nitem=3$ candidates with a completely uninformative ranking: $P[\perm^1] = P[\perm^2] = P[\perm^3]$, so the only source of information is the free busy status. Next, consider two status vectors: one is $[0, 0, 1]$ (candidates are busy, busy, free), and the second is $[1, 1, 0]$ (free, free, busy). We show that there exist parameters such that it is optimal to pick the \emph{last} candidate for both of these status vectors: this shows that, while for a fixed status vector the optimal strategy is always to pick the top-ranked free or top-ranked busy candidate, the optimal strategy cannot be described as in Algorithm \ref{algo:superstar} (\enquote{first-free} or \enquote{first-busy}) because firms may preferentially pick free or busy candidates based on the situation. 
\end{proof}

However, this result should not be very surprising: in the limit of an extremely uninformative ranking, we should not expect the optimal strategy to be well-described in terms of a \enquote{window} of top $j$ items in the ranking. Instead, we will show in Theorem \ref{thrm:vposntwoaccstrat} that Algorithm \ref{algo:superstar} has low additive error when the ranking tool is relatively accurate.

\begin{restatable}{theorem}{vposntwoaccstrat}\label{thrm:vposntwoaccstrat}
Algorithm \ref{algo:superstar} as applied to a setting where $j$ is the first index that differs from $1$ has error that is optimal up to an additive factor of $\epsilon$, for $\epsilon \leq \frac{p_2}{p_1} \cd \val_2 \cd (1-1/\busypen_2) \cd (1-P[\perm^1] - P[\perm^j])$. 
\end{restatable}
At a high level, the bound in Theorem \ref{thrm:vposntwoaccstrat} has high error when a) there is a relatively low probability of the high-valued candidate being in index $1$ or $j$ (the only two indices that Algorithm \ref{algo:superstar} considers choosing), or b) when the low value candidate has relatively high value. Both of these are scenarios where we expect preferentially picking candidates based on free-busy status to be a poor strategy. Figure \ref{fig:simsstrat} in the appendix includes numerical simulations studying the sub-optimality of Algorithm \ref{algo:superstar}, showing that often in practice the error is extremely small. Throughout the superstar analysis, we will assume that firms are acting according to Algorithm \ref{algo:superstar}.

\section{Superstar: welfare implications }\label{sec:welfare}

Section \ref{sec:optstrat} described the strategies that firms select under different parameter regimes: either preferentially looking for free candidates, for busy candidates, or simply following the ranking (picking the top ranked candidate). In this section, we will discuss how the firm's strategy impacts the broader system, such as the welfare of firms and candidates. Throughout, we will focus specifically on the impact of changing accuracy in the ranking tool on both strategy and welfare. 

\subsection{Welfare of firm}
First, we will note that welfare of the firm making hiring decisions always increases with increased accuracy of the ranking tool:
\begin{restatable}{theorem}{welfarefirm}\label{thrm:welfarefirm}
Increasing the accuracy always increases welfare of firms using a first-free or first-busy strategy as described in Section \ref{sec:optstrat}. 
\end{restatable}
This should be relatively intuitive: as the accuracy of the tool increases, the utility of a best-responding firm always increases. 

\subsection{Welfare for candidates:}\label{sec:welfarecandidates}

Next, we will consider impacts that the firm's hiring strategy has on candidate welfare. For example, a firm using the first-busy strategy will select busy candidates much more frequently than a firm using the first-free strategy. Picking a busy candidate could involve negative impacts, particularly in terms of equity: for example, it could be viewed as firms \enquote{poaching} candidates from other firms rather than hiring unemployed candidates, which may lead to a higher unemployment rate. However, picking a busy candidate could also be viewed positively: candidates who receive competing offers could potentially use them to increase their compensation.

What is the impact of increasing accuracy on how often busy candidates are picked? Theorem \ref{thrm:increaseddup} shows that this can \emph{increase} or \emph{decrease}.  For intuition, increasing the accuracy of the ranking tool has multiple opposing effects: first, it directly changes the distribution of candidates over indices, which in general means candidates ranked towards the top are more likely to be higher value. Secondly, because it changes the distribution of candidates, it also changes the frequency of status vectors: specifically, top-ranked candidates are more likely to be busy (because they are more likely to be high value candidates). Finally, it also changes the strategy of firms: because Algorithm \ref{algo:superstar} relies on the accuracy of the ranking tool, as the accuracy of this tool changes, the firm (hiring manager) will also update her strategy. 

Theorem \ref{thrm:increaseddup} considers all three of these simultaneous changes. Specifically, Theorem \ref{thrm:increaseddup} shows that, when accuracy increases, if a firm is using the strategy first-free, the probability that the picked candidate is busy always decreases. Conversely, if the firm is using the strategy first-busy, then the probability that the picked candidate is busy \emph{may increase} (if the firm's strategy window size stays constant) or \emph{may decrease} (if the window size shrinks).

\begin{restatable}{theorem}{increaseddupnew}\label{thrm:increaseddup}
Increasing the accuracy of the ranking tool \emph{increases} the probability that the picked candidate is busy among firms using first-free and \emph{may decrease} the probability that the picked candidate is busy among firms using first-busy.  
\end{restatable}

Figure \ref{fig:p_picked_free} illustrates this phenomenon with the RUM with Gumbel noise model. Specifically, both figures include colored, dashed lines showing the probability of a firm picking a busy candidate for every possible window length. Note that for both plots, as the accuracy increases, the colored lines increase: holding window length constant, increasing accuracy increases how frequently busy candidates are picked. For the left figure (first-free), as the window width increases, the probability of picking a busy candidate decreases, but for the right figure (first-busy), the probability decreases with window width: this is simply because shrinking the window size makes it less likely that each firm will find their preferred free/busy type within the window. The solid black lines in Figure \ref{fig:p_picked_free} shows how the firm's strategies change in response to the ranking tool's increase in accuracy. Recall from Lemma \ref{obs:vpos} that increased accuracy always shrinks window size. If a firm is using first-free, then increased accuracy causes the firm to move up and to the right along the colored lines, always increasing how often busy candidates are picked. However, if a firm is using first-busy, then increasing accuracy causes the firm's strategy to move down and to the right - and ultimately could increase or decrease the probability that a busy candidate is picked, depending on the degree to which the accuracy changes. 

\begin{figure}
\centering 
    \includegraphics[width=2.5in]{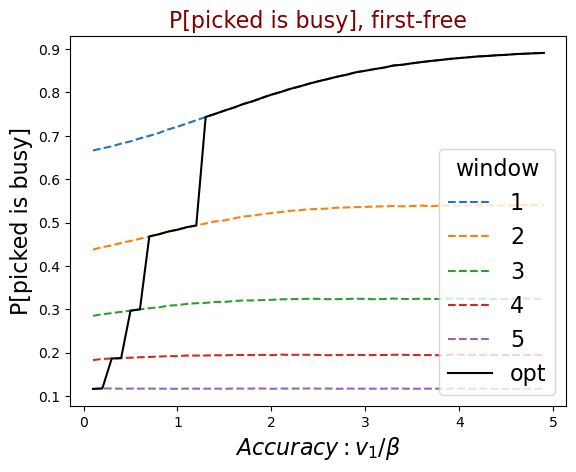}
    \includegraphics[width=2.5in]{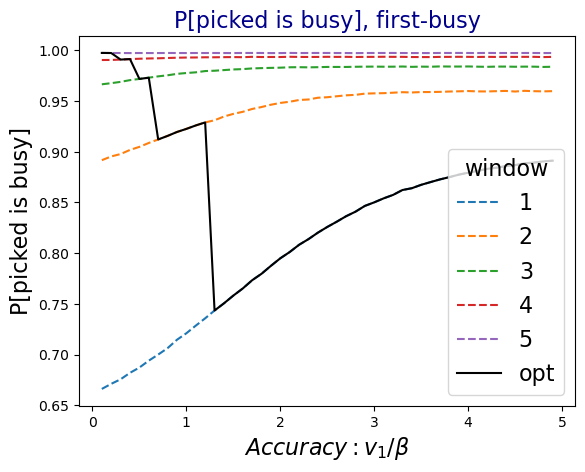}
    \caption{Figures showing how changing the accuracy of the tool changes how frequently the picked candidate is busy (fixed window sizes dashed, optimal window size solid). The top shows a firm using \enquote{first-free} strategy: note that increasing accuracy uniformly increases the probability that the picked candidate is busy. The bottom figure shows a firm using first-busy strategy: note that increasing accuracy could increase or decrease the probability that the picked candidate is busy. Parameters: $\val_1 = 1, \val_2=0, p_1 =0.1, p_2 = 0.4$, each point run with $10^6$ simulations. First-free strategy run with $\busypen = 10$, first-busy run with $\busypen = 1.6$.}
    \label{fig:p_picked_free}
    \vspace{-0.5cm}
\end{figure}

\subsection{Strategic use of ranking tool (welfare of ranking company):}
Finally, we discuss the impact of increased accuracy on how often firms end up using the ranking tool as it was designed - specifically, on how often firms select the top-ranked candidate.

In light of Theorem~\ref{thrm:increaseddup}, one might be tempted to assume that increased accuracy leads to higher ``congestion at the top," in the form of firms being more likely to select highly-ranked candidates even if they are busy.  While it is true that this can occur, the relationship between accuracy and firm selection can also be more subtle: there are instances in which increasing the accuracy of the ranking tool actually \emph{decreases} the likelihood of picking the top-ranked candidate.

Such a non-monotonic relationship between tool accuracy and usage has potentially important implications for tool design.  For example, a firm that uses metrics like the click rate of top-ranked candidates might naturally (but incorrectly) infer from a reduced click rate that a supposed tool improvement led to reduced performance and welfare. Additionally, picking the top ranked item could be viewed as a direct positive for the ranking company, if, for example, they obtain higher advertising revenue when this occurs. The importance of studying how often the top candidate is picked is especially salient when increasing the accuracy of the ranking tool involves a heavy investment in increased data collection, increased algorithmic sophistication, or increased engineering effort. However, as we now show, an improved ranking tool can provide improved welfare for the firm (as guaranteed by Theorem~\ref{thrm:welfarefirm}) while simultaneously causing users to select top-ranked candidates less frequently. 
Specifically, Theorem \ref{thrm:increasedacctopclicknew} shows that increasing the accuracy of the ranking tool can either \emph{increase} or \emph{decrease} the probability that a firm selects the top ranked candidate. 

\begin{restatable}{theorem}{increasedacctopclicknew}\label{thrm:increasedacctopclicknew}
Increasing the accuracy of the ranking tool could \emph{increase} or \emph{decrease} how frequently the firm picks the top-ranked candidate. 
\end{restatable}

The full proof gives an exact condition for when increasing accuracy increases the chance that the top candidate is picked. For intuition, a smaller window size always \emph{increases} the chances that the top candidate is picked: this should make intuitive sense, as the firm is considering fewer alternatives. However, holding window size constant, increasing accuracy could \emph{increase} or \emph{decrease} the probability that the top candidate is picked. Generally, if the firm is using a first-busy strategy, increasing accuracy increases the probability that the top candidate is busy, and then increases the chance that that candidate is picked. Conversely, if the firm is using a first-free strategy, increasing accuracy generally decreases the probability that the firm picks the top ranked candidate because the firm is actively avoiding busy candidates\footnote{Interestingly, the full story is somewhat more nuanced than this: there exist cases where increasing accuracy can \emph{decrease} the probability the top candidate is picked with first-busy, and \emph{increase} the probability with first-free, typically with extreme values of $p_2$ and very inaccurate rankings: see Appendix \ref{app:welfare} for more analysis.}. Figure \ref{fig:increasedecreasetop} in the Appendix gives a corresponding figure to Figure \ref{fig:p_picked_free} illustrating how the probability of the top-ranked item being picked could increase or decrease with accuracy. At a high level, Theorem \ref{thrm:increasedacctopclicknew} illustrates how, at the most basic level, increased accuracy of a ranking tool could lead to counterintuitive effects around how often firms end up picking the candidate suggested by the ranking. 

\section{Beyond superstar}\label{sec:beyondsuperstar}

In this section, we further relax the model and consider cases where the $\nitem$ items can come in more than two types of values and probabilities of being available. First, we show that very similar welfare properties hold even in the beyond-the-superstar setting, assuming that firms use a strategy similar to \enquote{first-free/first-busy} where they create a window size (shrinking with accuracy) and preferentially pick either the top-ranked free or busy candidate within that window. Secondly, we analyze \emph{when} such a strategy is approximately optimal. While we show that the strategy space is much more complex in the beyond-the-superstar setting, we show that under reasonable noise models, a similar strategy achieves low expected error.

\subsection{Welfare}\label{sec:beyondwelfare}
First, we will discuss the impacts of increased accuracy on measures of societal welfare, focusing on extending results from Section \ref{sec:welfare}. Throughout, we will assume that firms are using first-free or first-busy: specifically, that firms are preferentially picking either free or busy candidates within some window, and that the size of this window only shrinks with increased accuracy. Note that this is a very general condition, and we believe that this encompasses a wide range of natural strategies by the firm. 

First, Theorem \ref{thrm:beyondwelfare} generalizes Theorem \ref{thrm:increaseddup}, giving conditions for when the probability of the picked candidate being busy will increase or decrease as accuracy of the permutation increases:  

\begin{restatable}{theorem}{beyondwelfare}\label{thrm:beyondwelfare}
For the beyond-the-superstar setting, if the firm is using a first-free or first-busy strategy with a window size that weakly shrinks with increased accuracy, the following properties hold: 1) If the firm is using first-free, then increased accuracy increases the probability that the picked candidate is busy. 2) If the firm is using first-busy, then increased accuracy \emph{may decrease} the probability that the picked candidate is busy. 
\end{restatable}

Similarly, Theorem \ref{thrm:beyondpicktop} generalizes Theorem \ref{thrm:increasedacctopclicknew}, showing that there exists cases where increasing the accuracy of the tool increases or decreases the probability that the top candidate is picked\footnote{This is directly implied by Theorem \ref{thrm:increasedacctopclicknew} because the superstar case is a subset of the beyond-the-superstar case, but the proof of Theorem \ref{thrm:beyondpicktop} gives more general conditions for when we expect the probability to increase or decrease.}. 

\begin{restatable}{theorem}{beyondpicktop}\label{thrm:beyondpicktop}
In the beyond the superstar setting, increasing the accuracy of the ranking tool could \emph{increase} or \emph{decrease} how frequently the firm picks the top-ranked candidate. 
\end{restatable}

At a high level, we view Theorems \ref{thrm:beyondwelfare} and \ref{thrm:beyondpicktop} as showing that increased accuracy could have counterintuitive or negative impacts on welfare, even in settings more general than the superstar model.  

\subsection{Strategy}\label{sec:beyondstratval}
Finally, we will turn to analyzing the firm's strategy in the beyond-the-superstar setting. As mentioned previously, this setting can be substantially more complex, but we will show that a similar first-free/first-busy strategy can achieve relatively low expected error.

\subsubsection{Restricting choice to either top-free or top-busy candidate}\label{sec:beyondstratvaltop}
First, recall that Lemma \ref{lem:alwaysborf} showed that in the superstar setting, it is always optimal to pick either the first free or first busy candidate (e.g., it is never optimal to pick the second free candidate). Theorem \ref{thrm:beyondsuperstartopbest} generalizes this to the beyond the superstar setting, again showing that for reasonable noise models it is always optimal to pick either the first free or first busy candidate: this dramatically reduces the strategy space for the firm. While this result is natural, proving it is surprisingly non-trivial. 

\begin{restatable}{theorem}{beyondsuperstartopbest}\label{thrm:beyondsuperstartopbest}
In the beyond the superstar setting, there exists probability distributions $\permdist$ such that for some realized status vector $\avail$, the optimal choice is to pick neither the first free nor the first busy candidate (e.g., instead to pick the 2nd busy candidate). However, for commonly-used noise models (e.g. Mallows model and Random Utility Model), it is always optimal to pick the first free or first busy candidate.  
\end{restatable}

Specifically, Theorem \ref{thrm:beyondsuperstartopbest} can be proven through a series of sub-lemmas. First, Lemma \ref{lem:counterexample} gives a specific example in the beyond-the-superstar setting where the optimal choice is to pick the \emph{second} free item. 

\begin{restatable}{lemma}{counterexample}
\label{lem:counterexample}
There exists a probability distribution $\permdist$ such that for some realized status vector $\avail$, picking the \emph{second} free item maximizes expected utility, even if $\permdist$ is descending in expected value. 
\end{restatable}

However, the example in Lemma \ref{lem:counterexample} is constructed by creating a somewhat \enquote{unnatural} permutation distribution with 0 probability on many permutations. Definition \ref{def:monotone} defines a natural property on permutations where, for every pair of permutations $\perm, \tilde \perm$ that differ solely in a pairwise inversion, $\perm$ is at least as likely as $\tilde \perm$.  

\begin{restatable}{definition}{monotonedef}\label{def:monotone}
    A probability distribution over permutations $\mathcal{P}$ is called \emph{inversion-monotone} if it satisfies the following condition: 
for any permutation $\sigma$ with a pair of indices $i<j$ such that $\val_{\sigma_i} > \val_{\sigma_j}$, we construct a corresponding permutation $\tilde \sigma$ where $\tilde\sigma = \sigma$ except for indices $i, j$, which we have flipped: $\val_{\tilde \sigma_i} = \val_{\sigma_j}, \val_{\tilde \sigma_j} = \val_{\sigma_i}$. 
Then, we require 
$P[\sigma] \geq P[\tilde \sigma]$.
\end{restatable}

Lemma \ref{lem:monotoneenough} shows why this property is important: any permutation distribution satisfying it must result in the optimal strategy being to pick either the first free or first busy candidate. 

\begin{restatable}{lemma}{monotoneenough}
\label{lem:monotoneenough}
Consider any vector of realized status vector $\avail$, with $\avail_i = \avail_j$, for some $i<j$. Then, if the probability distribution of permutations is \emph{inversion-monotone} as in Definition \ref{def:monotone}, the optimal solution will always be to pick the first free item or the first busy item. 
\end{restatable}

Finally, the proof of Theorem \ref{thrm:beyondsuperstartopbest} concludes by showing that commonly-used permutation distributions are inversion monotonic.  While we believe that this property may be of independent interest, to our knowledge, we are the first paper to prove such a property for the Random Utility model with Gaussian noise. 
\begin{restatable}{lemma}{bothmonotone}
\label{lem:bothmonotone}
The following permutation distributions are inversion monotone: 1) Mallows, 2) RUM with Gumbel noise, and 3) RUM with noise from any identical, symmetric noise distribution (e.g. Gaussian). 
\end{restatable}

\subsubsection{Algorithm for picking candidate}\label{sec:beyondstratvaltop}
Having established that the optimal choice is either the top-ranked free candidate or top-ranked busy candidate, we now turn to the question of how to pick between them. Algorithm \ref{algo:approximatebeyondsuperstar} generalizes Algorithm \ref{algo:superstar}, providing an algorithm for how to pick candidates in the beyond-the-superstar setting. 

\begin{algorithm}[htbp]
\caption{Given candidate values $\{\val_i\}$, probabilities of being free $\{p_i\}$, accuracy of ranking tool $\{P[\perm]\}$, $s \in \{0, 1\}$ for status of top-ranked candidate, and index $j$ of first candidate with different free-busy status from $s$. }
\label{algo:approximatebeyondsuperstar}
\DontPrintSemicolon
\LinesNumbered
\For{$v_i, v_k \mid v_i > v_k$ \Comment*[r]{Go over all ordered pairs}}{
    Define $\ratio_{i, k} = \pfree_k \cd (1-\pfree_i)/((1-\pfree_k) \cd \pfree_i)$. \\
    Define $\perm_{i, k, j}$ as permutation with every item ordered correctly, but with $\val_i$ ranked 1st and $\val_k$ is ranked $j$th. \\
    Define $\tilde \perm_{i, k, j}  = \perm_{i, k, j}$ except with indices $1, j$ flipped.\\
    \eIf{$s=0$ \Comment*[r]{If top candidate is busy. }}{
    $G_{ik} = \ratio_{i, k} \cd P[\perm] \cd (\val_{k} - \val_{i}/\busypen_{i}) + P[\tilde \perm] \cd (\val_{i} - \val_{k}/\busypen_{k})$ \\
    }{\Comment*[r]{If top candidate is free}
    $G_{ik} = P[\perm] \cd (\val_{k}/\busypen_{k} - \val_{i}) + \ratio_{i, k} \cd P[\tilde \perm] \cd (\val_{i}/\busypen_{i} - \val_{k})$ \\
    }
}
    $\mathcal{G}_1 = \{G_{ik} \mid G_{ik} \leq0\}$, 
    $\mathcal{G}_j = \{G_{ik} \mid G_{ik} >0\}$\\
\eIf{$\abs{\mathcal{G}_j} > \abs{\mathcal{G}_1}$ \Comment*[r]{Vote on index to pick}}{ 
\Return{$j$}
}{
\Return{$1$}
}
\end{algorithm}
At a high level, Algorithm \ref{algo:approximatebeyondsuperstar} works by iterating over all pairs of candidates and allowing them to \enquote{vote} about whether it would be better to pick the candidate ranked in index $1$ or in index $j$. The ultimate choice comes down to what the majority of pairs decide. 
Next, we will prove several properties about Algorithm \ref{algo:approximatebeyondsuperstar}. First, Lemma \ref{lem:beyondreduces} shows that Algorithm \ref{algo:approximatebeyondsuperstar} reduces to Algorithm \ref{algo:superstar} in the superstar setting, as desired. 

\begin{restatable}{lemma}{beyondreduces}\label{lem:beyondreduces}
In the superstar setting, Algorithm \ref{algo:approximatebeyondsuperstar} becomes Algorithm \ref{algo:superstar}.
\end{restatable}

Next, Lemma \ref{lem:beyondprop} shows that Algorithm \ref{algo:approximatebeyondsuperstar} is actually also a $k$-free/$k$-busy strategy with a window that shrinks with increased accuracy, exactly as Algorithm \ref{algo:superstar} is. This result is especially important, because it implies that Theorem \ref{thrm:beyondwelfare} applies, meaning that any firm using the strategy in Algorithm \ref{algo:approximatebeyondsuperstar} will result in welfare implications that match those in the superstar setting. 

\begin{restatable}{lemma}{beyondprop}\label{lem:beyondprop}
Given a constant $\busypen$, Algorithm \ref{algo:approximatebeyondsuperstar} is a $k$-free/$k$-busy strategy: it can be reduced to picking a window size and preferentially picking the top-ranked free or busy candidate within that window. The window size shrinks with increased accuracy, and if the firm is preferentially picking busy candidates, the window shrinks with increased $\busypen$ (increasing if the firm is preferentially picking free candidates).
\end{restatable}

Finally, Lemma \ref{lem:beyondstratopterr} analyzes the sub-optimality of Algorithm \ref{algo:approximatebeyondsuperstar}. While Algorithm \ref{algo:approximatebeyondsuperstar} does satisfy several nice properties (as in Lemma \ref{lem:beyondprop}), it is an approximation to directly calculating the optimal posterior choice, and as such as some degree of error. In particular, two sources of error are 1) making a decision based on majority \enquote{vote} across pairs of candidates, which leads to error from candidates who \enquote{voted} the opposite direction, and 2) the approximation in $G_{ik}$, which uses permutations $\perm, \tilde \perm$ that make it as easy as possible to select the top candidate.  

\begin{restatable}{lemma}{beyondstratopterr}\label{lem:beyondstratopterr}
For a particular index $1, j$, denote $\mathcal{G}_+$ as $\{v_i, v_j\mid G_{i, k} > 0\}$, $\mathcal{G}_-$ as $\{v_i, v_j\mid G_{i, k} < 0\}$, $\mathcal{I}_{1j}(i,k)$ as the event that items $\{i, k\}$ are in indices $\{1, j\}$, and $\tilde G_{i, k}$ as $G_{ik}$ calculated in Algorithm \ref{algo:approximatebeyondsuperstar}, but with $\perm_{i, k, j}', \tilde \perm_{i, k, j}'$ instead defined as permutations where every item is \emph{incorrectly} ordered. Then, if Algorithm \ref{algo:approximatebeyondsuperstar} picks a candidate in index $j$, the error is upper bounded by: 
$$\sum_{i, k \in \mathcal{G}_-} P[\mathcal{I}_{1j}(i,k)] \cd (\val_i - \val_k) $$
and if Algorithm \ref{algo:approximatebeyondsuperstar} picks a candidate in index $1$, the error is upper bounded by: 
$$\sum_{i, k \in \mathcal{G}_+} P[\mathcal{I}_{1j}(i,k)] \cd (\val_i - \val_k) + \sum_{i, k \in \mathcal{G}_-}P[\mathcal{I}_{1j}(i,k)] \cd \max(-\tilde G_{i, k}, 0) $$
\end{restatable}

Note that the error of this algorithm is high when there is a large chance of observing candidates in indices $1, j$ where the relative benefits of being busy versus free are the \emph{opposite} of candidates in general. However, this error bound is very natural: in such a setting, there is no consistent signal in value related to a candidate being free or busy, and thus we would not expect any algorithm preferentially selecting candidates based on status to have good performance. Figure \ref{fig:simsstrat} (bottom) in the Appendix illustrates the expected error of Algorithm \ref{algo:approximatebeyondsuperstar} in the beyond-the-superstar setting, showing that in many parameter regimes, Algorithm \ref{algo:approximatebeyondsuperstar} has low or 0 error. 

\section{Discussion}\label{sec:discuss}

In this paper, we have proposed a simple model to study cases where strategic, self-interested agents use algorithmically generated rankings in the presence of side information, specifically a free-busy status vector for different candidates. Our results show that even this relatively simple model can show surprisingly complex phenomena.

\noindent \textbf{Firm's strategy:} First, we began by analyzing the superstar model, where a single candidate has higher value to the company than all other candidates. In Section \ref{sec:optstrat} we show that the firm's strategy can admit a surprisingly simple form: decide whether to preferentially seek out free or busy candidates, and then decide how far down the list to look for them based on the accuracy of the ranking tool. As the accuracy of the ranking tool increases, the firm only shrinks its \enquote{window size} and the firm's welfare only increases. 

\noindent \textbf{Welfare of candidates and ranking company:} When firms select preferentially based on the free-busy (e.g. employment) status of candidates, this immediately leads to differences in selection rate based on employment status. Preferentially selecting busy candidates or free candidates could have positive or negative effects (e.g. impacts on unemployment rate or competing offers). In Section \ref{sec:welfare}, we showed that there are scenarios where increasing the accuracy of the ranking tool could \emph{increase} how frequently the firm selects a busy candidate. 

Often the incentives of ranking companies, their clients, and societal welfare are assumed to be at least weakly aligned: producing higher-quality ranking tools should increase the value to both companies and clients. However, our results show that their incentives may be misaligned: in particular, a more accurate ranking tool can reduce measures of social welfare and fairness, and can even reduce the welfare of the ranking company itself by decreasing how frequently the top candidate is selected. Our results suggest that improving overall welfare may require more precise interventions than simply increasing the accuracy of tools. 

\noindent \textbf{Beyond superstar:} Finally, in Section \ref{sec:beyondsuperstar}, we discuss how our results generalize beyond the Superstar model. Specifically, we show that reasoning about the strategy of agents using ranking tools becomes substantially more subtle, and we prove properties about permutation models that may be of independent interest. However, given reasonable assumptions, we obtain similar results for strategy and welfare as in the Superstar setting, suggesting that our results generalize beyond a specific model.

\noindent \textbf{Implications:} While we focus on a specific motivating example of an agent with side information using an algorithmic ranking tool, more broadly we are interested in how algorithmic tools are used as part of broader systems. In general, many algorithmic tools are used not in isolation, but in conjunction with other systems or sources of information. Our motivation in this paper is to show that even in a relatively simple human-algorithmic system, increased algorithmic accuracy can have non-intuitive impacts. In general, we wish to illustrate the importance of understanding such impacts and motivate wider study of such settings. 

\noindent \textbf{Extensions: } There are multiple fascinating extensions to our model. For example, in this work we assumed that job candidates are honest about their free-busy status, but if firms preferentially select candidates based on this status, job candidates may be incentivized to mis-represent whether or not they are free. An interesting direction for future work would be to explore the extent to which such settings would impact welfare of firms and candidates. Separately, one policy intervention may be to deliberately hide their free/busy status from companies: this could potentially be in pursuit of increasing fairness in selection rates among candidates with different free/busy status. If firms are aware that busy candidates are less likely to accept offers, but are unaware of which candidates are busy, the optimal strategy of how firms select candidates changes. Additionally, these results generally assume that firms perfectly know certain parameters such as the $\ratio, \busypen$ (which could be viewed as measuring the efficiency of the market and the attractiveness of any given firm, respectively). In reality, firms only have access to estimates of these values, and in particular they may systematically mis-estimate these values, for example, assuming that their firm is more attractive than it really is. In this case, firms may use sub-optimal strategies, which may lead to decreased welfare for multiple agents. Studying how each of these extensions change the overall relationship between increased accuracy and welfare would be interesting avenues for future work. 

\subsection*{Acknowledgments}
This work was partially supported by a MIT METEOR Fellowship, and was partially completed while Kate Donahue was at an internship at Microsoft Research. We wish to thank the following individuals and groups for their very helpful comments: Yeganeh Alimohammadi, Keegan Harris, John Horton, Meena Jagadeesen, Sloan Neitert, Manish Raghavan, Michelle Si, Kiran Tomlinson, and Nicolas Wu, and attendees of the following events: Boston College Computer Science Colloquium, Simons Collaboration on the Theory of Algorithmic Fairness Annual Meeting, Themis Science Talk at Amazon, the AI and Science Group at Meta, Harvard Economics and Computer Science Colloquium, INFORMS Annual Meeting, and the AImpactCenter at UIUC for their valuable discussions and insightful suggestions.
We also thank the anonymous reviewers for their constructive feedback.

\clearpage 
\bibliographystyle{plain}
\bibliography{sample-bibliography}

\clearpage

\appendix
\begin{figure} 
\centering 
    \includegraphics[width=3in]{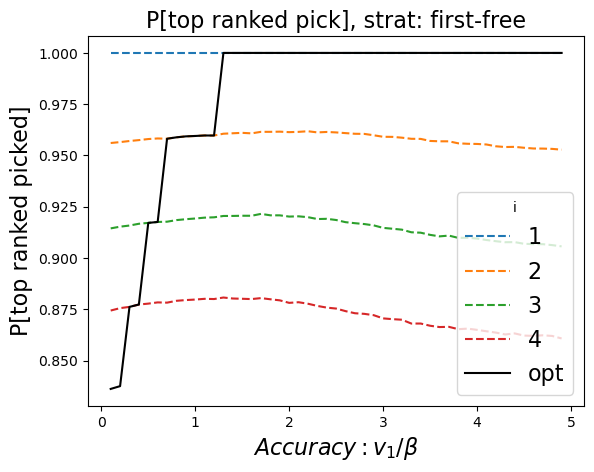}
    \includegraphics[width=3in]{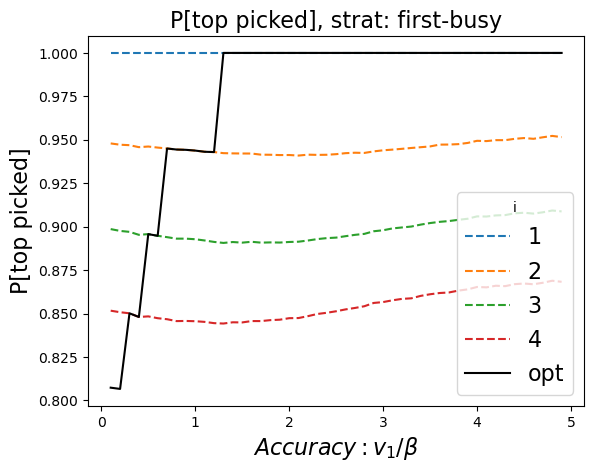}
    \caption{Figure showing examples where increasing accuracy could lead to \emph{increased} probability of picking the top ranked candidate for first-free strategies (top figure) or decreased for first-busy (bottom). In both cases, the phenomenon occurs for very low accuracy (far left part of plot). Both figures show $\nitem=10$ with only lines shown for $i\in [1, 4]$ to show detail, with RUM with Gumbel noise. Top has $p_1 = 0.01< p_2 = 0.05$, bottom has $p_1 = 0.9, p_2 = 0.95$. Numerical simulations with $10^6$ simulations per point. }\label{fig:increasedecreasetop}
\end{figure}
\begin{figure}
    \centering
    \includegraphics[width=0.6\linewidth]{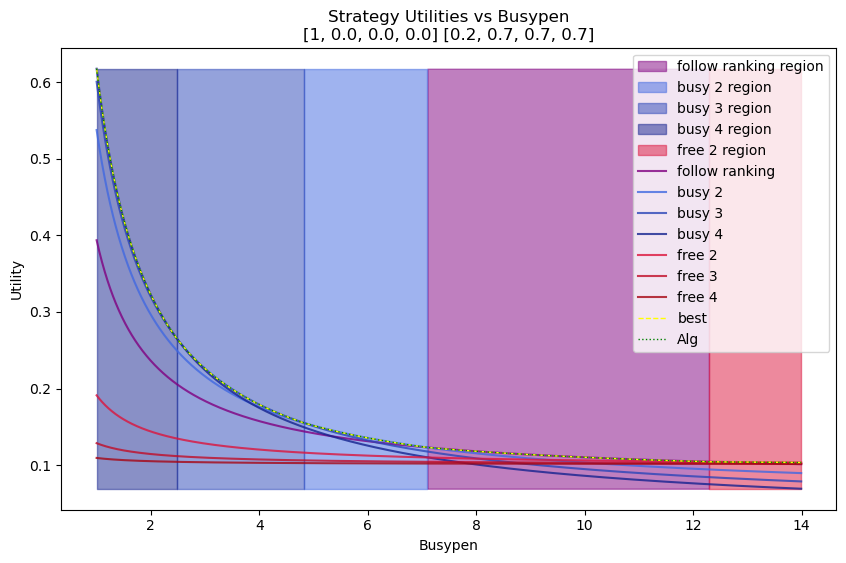}
    \includegraphics[width=0.6\linewidth]{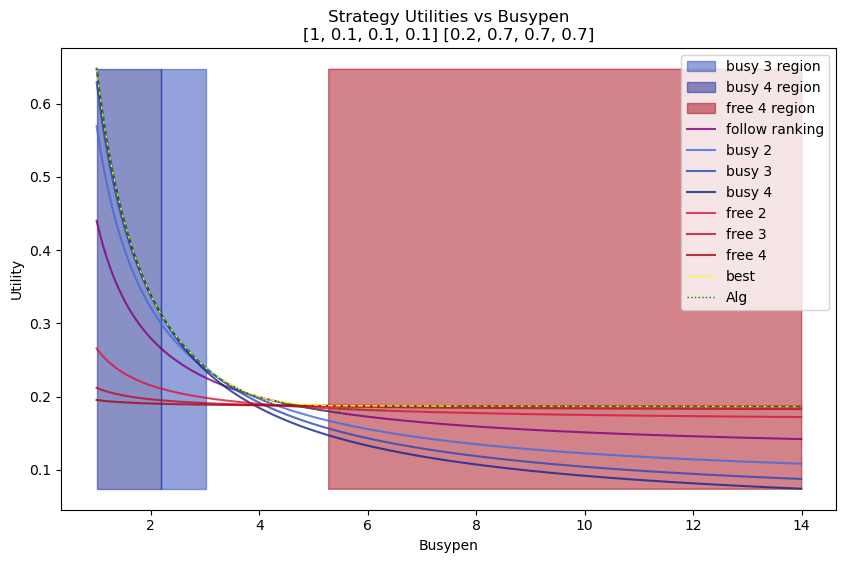}
    \includegraphics[width=0.6\linewidth]{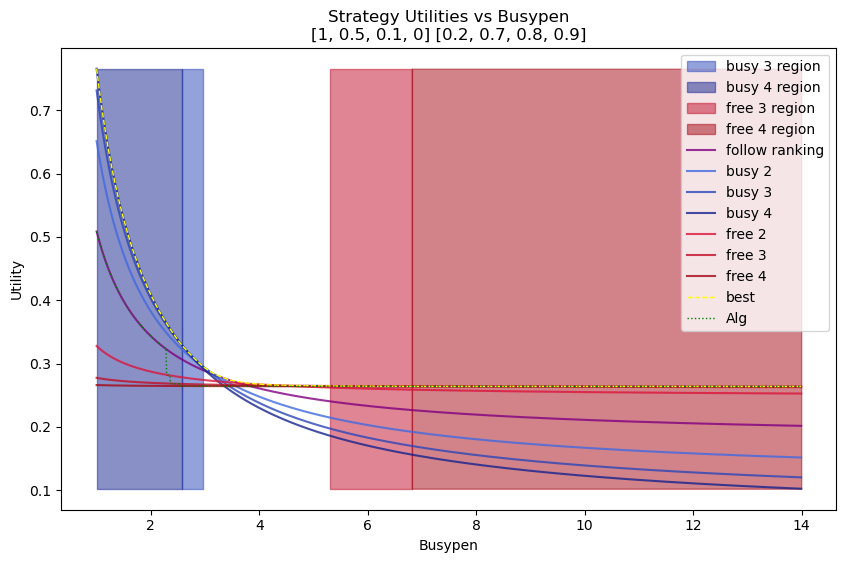}
    \caption{Simulation of performance of strategies with $\nitem=4$. Each line gives the expected utility of different strategies, where \enquote{best} gives the Bayes-optimal posterior strategy and Alg gives the performance of Algorithm \ref{algo:approximatebeyondsuperstar} (which reduces to Algorithm \ref{algo:superstar} in the superstar setting). The shaded regions illustrate where different named strategies happen to also be optimal: for example, the purple region is one where \enquote{follow the ranking} is exactly optimal. Unshaded regions are those where the optimal strategy cannot be described as $k$-free or $k$-busy. However, note that the error (gap between yellow and green lines) is often extremely small. The top and middle are both superstar settings, while the bottom is beyond-the-superstar. Note that the top has 0 error always because $\val_2 = 0$ (Theorem \ref{thrm:vposntwoaccstrat}), while the middle has nonzero error. The bottom (beyond-the-superstar) has regions with non-zero error, but largely follows the optimal strategy.}
    \label{fig:simsstrat}
\end{figure}

\newpage
\onecolumn
\section{Numerical extensions}\label{app:sim}
In Section \ref{sec:modelassump}, we presented our theoretical model. Here, we numerically explore a setting endogenizing our model. Specifically, the setting operates as follows: 
\begin{itemize}
    \item At time $t=0$, all candidates are free. 
    \item At each time step, a firm obtains a noisy permutation over candidates, governed by candidate values. Suppose that the firm selects the highest-ranked free candidate. 
    \item Every busy candidate has some probability of (independently) becoming free. 
\end{itemize}
Given this setting, high-value candidates will be more likely to be at the top of the ranking, and thus more likely to be busy. Depending on the parameters of the setting, it may be beneficial for firms to instead follow other strategies, such as those in Algorithm \ref{algo:superstar} or Algorithm \ref{algo:approximatebeyondsuperstar}. 

In Figures \ref{fig:sim1} and \ref{fig:sim2} we explore these settings. Note that these figures differ solely in the busy penalty $\busypen$. In both, the probability of observing the high value candidate in index $i$ is given by: 
$$[0.3704, 0.2507, 0.1632, 0.1011, 0.0589, 0.0317, 0.0153, 0.0063, 0.002, 0.0004]$$
Which means that the normalized version  $P[\perm^1]/P[\perm^i]$ is given by:
$$[1.477, 2.27, 3.663, 6.285, 11.682, 24.225, 58.902, 185.379, 981.489]$$
Given these parameters, if at each time step a firm picks the top-ranked free candidate, then Figure \ref{fig:sim0} shows that, on average, the probability that the high value candidate is free is around 0.37 and the probability that any given low value candidate is free is around 0.8. Given these parameters, $\ratio  = \frac{p_2/(1-p_2)}{p_1/(1-p_1)} \approx 6.5$. In order to calculate the \enquote{optimal} strategy according to Algorithm \ref{algo:superstar}, we only need to know $\busypen$.  

\begin{figure}
\centering
    \includegraphics[width=0.6\linewidth]{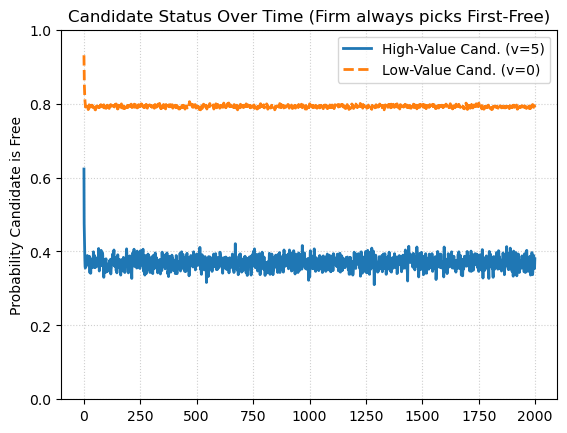}
    \caption{Simulation with $T=2000$ time steps, averaged over 1000 simulations, showing the average probability that the high value and low value candidates are free. }
    \label{fig:sim0}
\end{figure}

For $\busypen = 1.5$, as in Figure \ref{fig:sim1}, we have $R = 4.3$. By Algorithm \ref{algo:superstar}, this means that the firm should follow \enquote{first-busy}, and given the normalized values $P[\perm^1]/P[\perm^i]$, the window size should be $k=4$ (because $k=4$ is the maximum index such that $P[\perm^1]/P[\perm^i] < 4.3$. Figure \ref{fig:sim1} numerically simulates the expected utility for every top-k strategy, both first-free and first-busy, showing that first-busy with window size $k=3$ is, in fact, optimal!

\begin{figure}
    \includegraphics[width=0.5\linewidth]{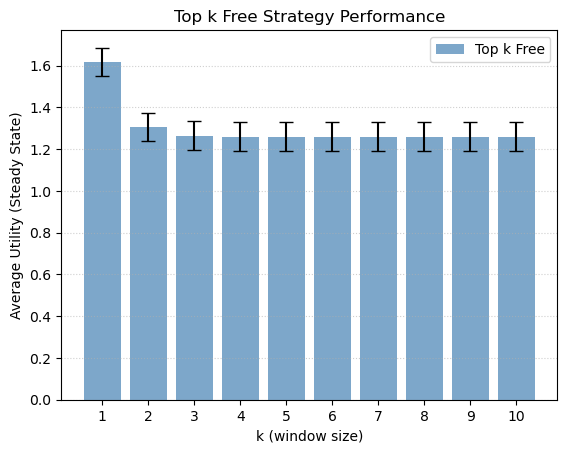}
    \includegraphics[width=0.5\linewidth]{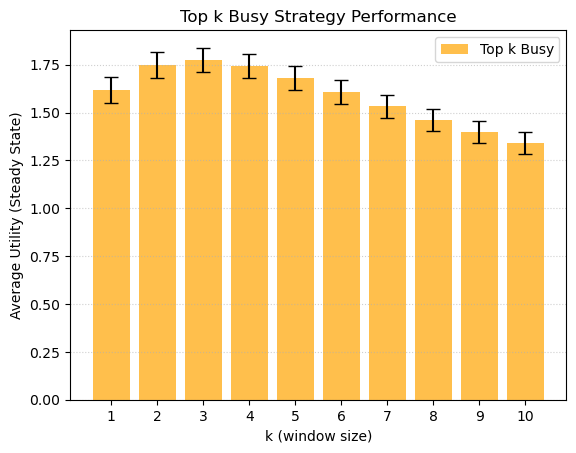}
    \caption{Simulation with $\nitem =10$ candidates, in superstar setting with $\val_1 = 5 > \val_2 = 0$. Ranking given by Plackett-Luce with $\beta = 3$. Each candidate has probability 0.4 independently of becoming free, if busy. Penalty for picking a busy candidate is $\busypen = 1.5$. }
    \label{fig:sim1}
\end{figure}

Next, for $\busypen = 15$, as in Figure \ref{fig:sim2}, we end up with $R =0.43$, which by Algorithm \ref{algo:superstar} results in a first-free strategy with window size $k=2$. Figure \ref{fig:sim2} shows that for this setting, the simulation similarly shows that this strategy is optimal.  

\begin{figure}
    \includegraphics[width=0.5\linewidth]{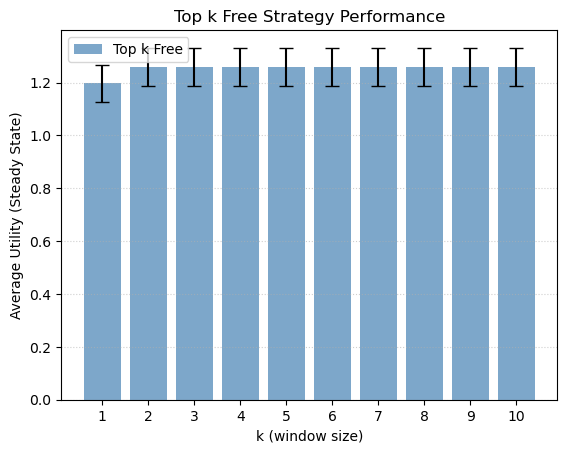}
    \includegraphics[width=0.5\linewidth]{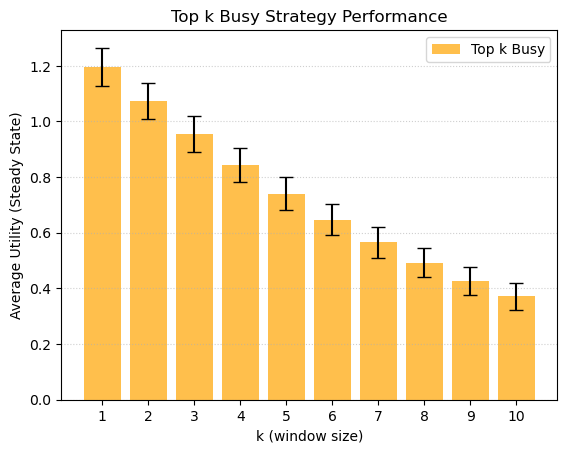}
    \caption{Same as Figure \ref{fig:sim1}, but with $\busypen = 15$. }
    \label{fig:sim2}
\end{figure}

Finally, for $\busypen = 4$, as in Figure \ref{fig:sim3}, we end up with $R =1.63$, which by Algorithm \ref{algo:superstar} results in a first-busy strategy with window size $k=1$ - which is equivalent to \enquote{follow-the-ranking}. Figure \ref{fig:sim3} shows that for this setting, the simulation similarly shows that this strategy is optimal.

\begin{figure}
    \includegraphics[width=0.5\linewidth]{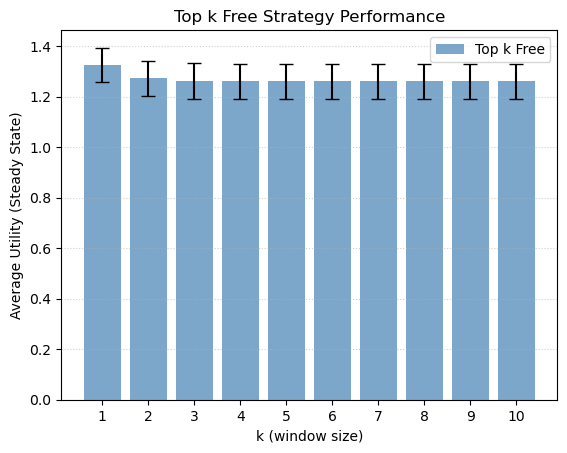}
    \includegraphics[width=0.5\linewidth]{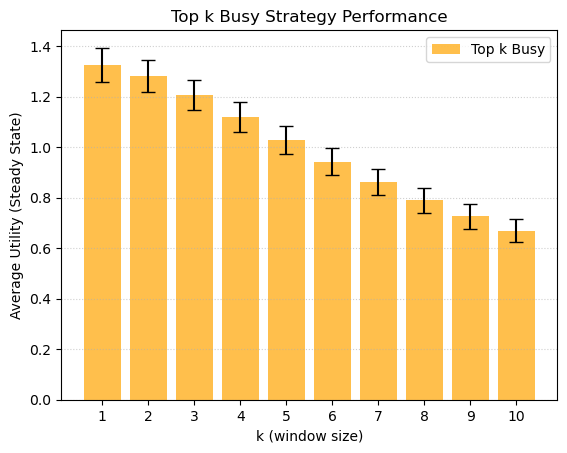}
    \caption{Same as Figure \ref{fig:sim1}, but with $\busypen = 4$. }
    \label{fig:sim3}
\end{figure}

\subsection{Variant with competing strategic firms}
In the previous section, we studied the best-response strategy, assuming that existing firms are preferentially picking the highest-ranked free candidate. However, what happens if existing firms are also picking strategically? Specifically, consider this variant: 

\begin{itemize}
    \item At time $t=0$, all candidates are free. 
    \item At each time step, a firm obtains a noisy permutation over candidates, governed by candidate values. \textbf{Suppose that the firm selects the highest-ranked busy candidate with a $k=3$ window size. }
    \item Every busy candidate has some probability of (independently) becoming free. 
\end{itemize}

In this setting, the existing firm is using the optimal strategy for $\busypen = 1.5$, as in Figure \ref{fig:sim1}. How does this change the optimal best-response strategy? 

Figure \ref{fig:sim4} shows the results of this simulation. Note first that the probability of the high value candidate being free drops substantially, to around 0.25, and the probability of low value candidates increases, to around 0.9. This is because having the other firm preferentially pick busy candidates makes it more likely that the high value candidate is busy. Because these probabilities change, this also changes the best-response strategy. Note that in the lowermost plot, the strategy that maximizes the welfare of the best-responding firm is still first-busy, but now with a larger window size (roughly $k=5$). At a high level, because there's a stronger correlation between value and free/busy status, it becomes more \enquote{worth it} for firms to hunt for busy candidates, even if they're further down the ranking (placing less weight on the algorithmic tool). This could be viewed as a motivation for herding. 

\begin{figure}
\centering 
    \includegraphics[width=0.5\linewidth]{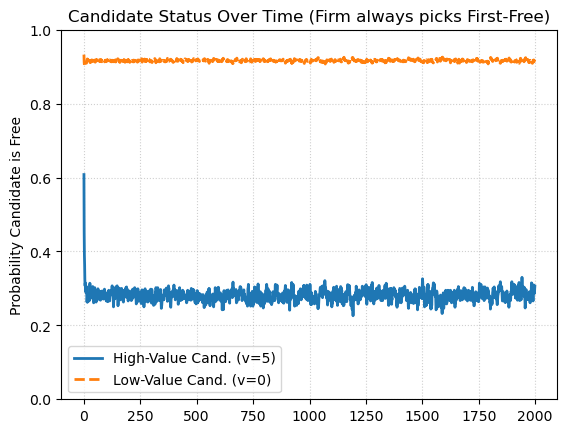}
    \includegraphics[width=0.5\linewidth]{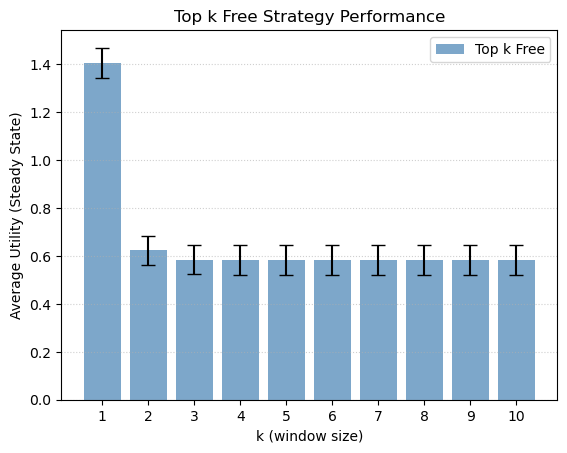}
    \includegraphics[width=0.5\linewidth]{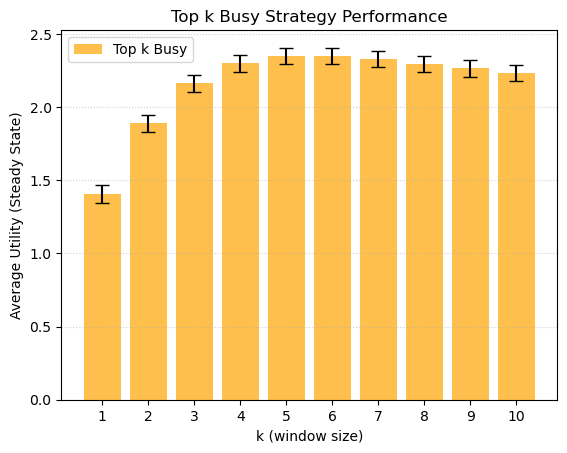}
    \caption{Same as Figure \ref{fig:sim1}, but where existing firm is picking first-busy with window size. }
    \label{fig:sim4}
\end{figure}

\section{Proofs for Section \ref{sec:optstrat}}\label{app:optstrat}
First, we will present general results on properties of distributions over rankings (\Cref{app:superstarpermprop}) and then prove results given in Section \ref{sec:optstrat}. 

\subsection{Properties of ranking permutations}\label{app:superstarpermprop}

\begin{restatable}{definition}{accuracy}\label{def:accuracy}[Accuracy: superstar]
A ranking $\{P[\perm^i]\}'$ is \emph{more accurate} than $\{P[\perm^i]\}$ if both conditions hold: 
\begin{enumerate}
     \item Ratio of probabilities: The relative probability of the high value candidate being ranked more highly always increases: $\frac{P[\perm^i]'}{P[\perm^j]'} \geq \frac{P[\perm^i]}{P[\perm^j]} \ \forall i< j$ with the inequality strict for at least one $i, j$ pair. 
   \item Monotonic increase: if the probability of the high value item being found at index $k$ decreases, then it also decreases for all $j>k$, or: if $P[\perm^k]'<P[\perm^k]$, then $P[\perm^j]'<P[\perm^j]$ for all $j>k$. Note that an immediate corollary of this property is majorization: $\sum_{i=1}^jP[\perm^i]' \geq \sum_{i=1}^jP[\perm^i]$ for all $j$. 
\end{enumerate}

\end{restatable}

\begin{restatable}{theorem}{gumbeltrick}\label{thrm:plackettluce}[From \cite{plackett1975analysis}, \cite{luce1959individual}]
Consider sets of items with values $v_i$, and a Random Utility Model with i.i.d. Gumbel noise $G(\mu, \beta)$ is added to each value. Then, the probability of observing the ordering $\perm$ is given by: 
$$\prod_{j=1}^{\nitem}\frac{\exp(\val_{\perm_j}/\beta)}{\sum_{\ell=j}^{\nitem}\exp(\val_{\perm_{\ell}}/\beta)}$$
\end{restatable}

\begin{restatable}{proposition}{gumbeltrickus}
\label{prop:gRUMbel}
Consider the superstar setting (1 item of value $v_1>v_2\geq 0$, and $n-1$ elements of item $\val_2$. If i.i.d. Gumbel noise (with scale parameter $\beta$) is added to each value, then the probability of the high value item being in index $i$ is given by: 
$$P[\perm^i] = \frac{\exp((\val_1-\val_2)/\beta)}{\exp((\val_1-\val_2)/\beta) + (n-1)} \cd \prod_{k=1}^{i-1} \frac{1}{1 + \frac{\exp((\val_1-\val_2)/\beta)-1}{n-k}}$$
\end{restatable}
\begin{proof}
We can prove this by relying on the Plackett-Luce model (Theorem \ref{thrm:plackettluce}), which says that the probability of observing any permutation $\perm$ is given by: 
$$\prod_{j=1}^{\nitem}\frac{\exp(\val_{\perm_j}/\beta)}{\sum_{\ell=j}^{\nitem}\exp(\val_{\perm_{\ell}}/\beta)}$$ 
Denote $\perm^k$ as a permutation with the high-value candidate in index $k$. Note that there are $(\nitem-1)!$ such permutations, because there are $(\nitem-1)!$ different ways to arrange the $\nitem-1$ low-value candidates. Applying Theorem \ref{thrm:plackettluce} in the superstar setting gives: 
$$\prod_{j=1}^{k-1}\frac{\exp(\val_2/\beta)}{\exp(\val_1/\beta) + (\nitem-j) \cd \exp(\val_2/\beta)} \cd \frac{\exp(\val_1/\beta)}{\exp(\val_1/\beta) + (\nitem-k) \cd \exp(\val_2/\beta)}\cd \prod_{j=k}^{\nitem-1}\frac{\exp(\val_2/\beta)}{\exp(\val_2/\beta) \cd (\nitem-j)}$$
$$=\prod_{j=1}^{k-1}\frac{1}{\exp((\val_1-\val_2)/\beta) + (\nitem-j)}  \cd \frac{\exp((\val_1-\val_2)/\beta)}{\exp((\val_1-\val_2)/\beta) + (\nitem-k)}\cd \frac{1}{\prod_{j=k}^{\nitem-1}(\nitem-j)}$$
$$ = \frac{\exp((\val_1 - \val_2)/\beta)}{\exp((\val_1 - \val_2)/\beta) + \nitem-1} \cd \prod_{j=2}^{k}\frac{1}{\exp((\val_1 - \val_2)/\beta) + \nitem-j} \cd \frac{1}{\prod_{j=k}^{\nitem-1}(\nitem-j)}$$
Reindexing: 
$$ = \frac{\exp((\val_1 - \val_2)/\beta)}{\exp((\val_1 - \val_2)/\beta) + \nitem-1} \cd \prod_{j=1}^{k-1}\frac{1}{\exp((\val_1 - \val_2)/\beta) + \nitem-j-1} \cd \frac{1}{\prod_{j=k}^{\nitem-1}(\nitem-j)}$$
Next, we multiply this by $(N-1)!$ to get the total probability $P[\perm^j]$ (probability of any permutation with high value item in index $j$). We note that: 
$$\frac{(\nitem-1)!}{\prod_{j=k}^{\nitem-1}(\nitem-j)}  = \prod_{j=1}^{k-1} (\nitem-j)$$

which means that the total probability can be written as
$$ \frac{\exp((\val_1 - \val_2)/\beta)}{\exp((\val_1 - \val_2)/\beta) + \nitem-1} \cd \prod_{j=1}^{k-1}\frac{\nitem-j}{\exp((\val_1 - \val_2)/\beta) + \nitem-j-1} $$
$$= \frac{\exp((\val_1 - \val_2)/\beta)}{\exp((\val_1 - \val_2)/\beta) + \nitem-1} \cd \prod_{j=1}^{k-1}\frac{1}{1 + \frac{\exp((\val_1 - \val_2)/\beta)-1}{\nitem-j}} $$
as desired. 
\end{proof}

From Proposition \ref{prop:gRUMbel} we can immediately observe that there is a multiplicative relationship between each of the permutations: 

\begin{observation}\label{obs:gRUMbel}
In the RUM with Gumbel noise and the superstar setting, for any $i<j$, we have: 
$$\frac{P[\perm^i]}{P[\perm^j]} = \prod_{k=i}^{j-1}\p{1 + \frac{\exp((\val_1-\val_2)/\beta)-1}{n-k}}$$ 
\end{observation}

\begin{restatable}{corollary}{gRUMbelacc}\label{cor:gRUMbelacc}
For the RUM with Gumbel noise, decreasing $\beta$ induces a distribution $\{P[\perm^i]'\}$ that is more accurate than $\{P[\perm^i]\}$ according to Definition \ref{def:accuracy}. 
\end{restatable}
\begin{proof}
First, we will prove the first property: that the relative probability of the high value candidate being ranked more highly always increases: $\frac{P[\perm^i]'}{P[\perm^j]'} \geq \frac{P[\perm^i]}{P[\perm^j]} \ \forall i< j$ with the inequality strict for at least one $i, j$ pair. This comes almost immediately from Observation \ref{obs:gRUMbel}. Note that increasing $(\val_1-\val_2)/\beta$ causes the term within Observation \ref{obs:gRUMbel} to increase, thus increasing the ratio of probabilities $P[\perm^i]/P[\perm^j]$. 

Second, we will prove the second property: if the probability of the high value item being found at index $k$ decreases, then it also decreases for all $j>k$, or: if $P[\perm^k]'<P[\perm^k]$, then $P[\perm^j]'<P[\perm^j]$ for all $j>k$. 

First, we note that from Proposition \ref{prop:gRUMbel} we can write: 
$$P[\perm^i] = P[\perm^1] \cd \prod_{k=1}^{i-1}q_k$$
for $q_k<1$, with $q_k$ shrinking with increased accuracy ($(\val_1 - \val_2)/\beta)$). We can write: 
$$P[\perm^i]' - P[\perm^i] = P[\perm^1]' \cd \prod_{k=1}^{i-1}q_k' - P[\perm^1] \cd \prod_{k=1}^{i-1}q_k$$
Suppose that $P[\perm^i]' - P[\perm^i]<0$ for some $i$: we will aim to show that this must also hold for $i+1$. 

We can observe that the model in Proposition \ref{prop:gRUMbel} satisfies $P[\perm^1]' > P[\perm^1]$, so it must be that $\prod_{k=1}^{i-1}q_k' < \prod_{k=1}^{i-1}q_k$. Because $q_i'\leq q_i < 1$, we must have that $\prod_{k=1}^{i}q_k' < \prod_{k=1}^{i}q_k$, which implies that 
$$P[\perm^1]' \cd \prod_{k=1}^{i-1}q_k' - P[\perm^1] \cd \prod_{k=1}^{i-1}q_k = P[\perm^i]' - P[\perm^i] <0$$
as desired. 
\end{proof}

The most common noise distribution with the Random Utility Model is Gaussian noise, while our theoretical results (Proposition \ref{prop:gRUMbel}) uses Gumbel noise as in the Placket-Luce distribution (\cite{plackett1975analysis, luce1959individual}). Here, we show that empirical results for Gaussian noise closely mimic theoretical results for Gumbel noise. Figure \ref{fig:pdfs} shows PDFS for Gumbel and Gaussian distributions, Figure \ref{fig:psis} shows the $P[\perm^i]$ values for the superstar model, and Figure \ref{fig:ps1psis} shows the ratio $P[\sigma^1]/P[\sigma^i]$, both for the Gumbel and Gaussian noise distribution. In all settings, the $\beta$ and standard deviation values are set to be equal. 

\begin{figure}
\centering 
    \includegraphics[width=2.5in]{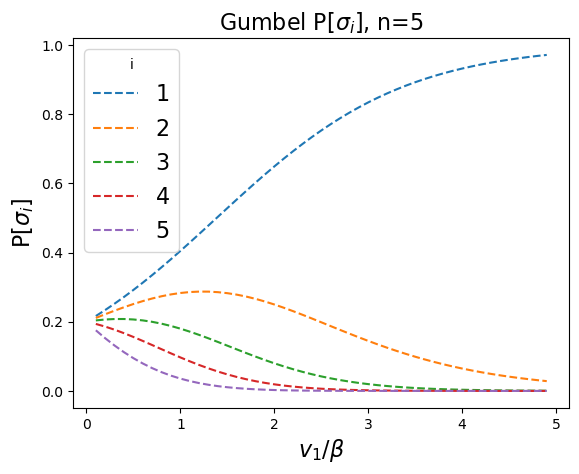}
    \includegraphics[width=2.5in]{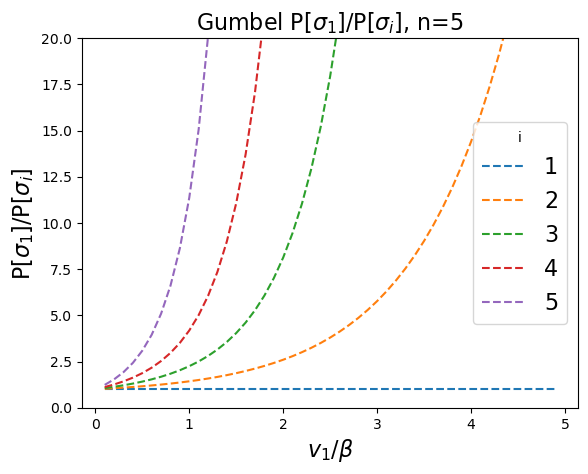}
    \caption{Left: Figure showing the probability of the high-value item being in index $i$, for $i \in [1, 5]$ and $n=5$, as a function of accuracy with the RUM with Gumbel noise. Note that if the line decreases for $i$, it also decreases for $j>i$ (second part of Definition \ref{def:accuracy}). Right: Figure showing ratio of probability of high value item being ranked first over probability of it being in index $i$. Note that this ratio always increases (first part of Definition \ref{def:accuracy}).}
\end{figure}

\begin{figure}
\centering 
    \includegraphics[width=3in]{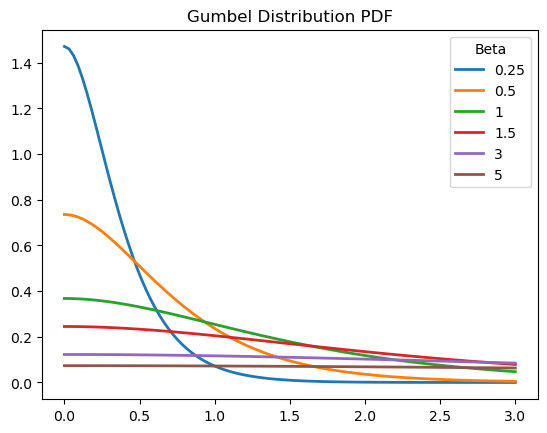}
    \includegraphics[width=3in]{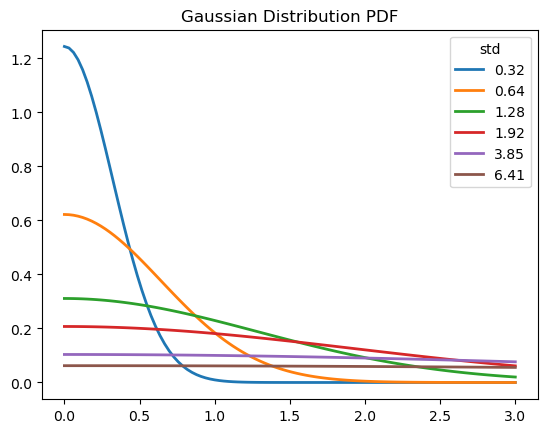}
    \caption{Figure of PDFs of Gumbel and Gaussian noise, respectively.}\label{fig:pdfs}
\end{figure}

\begin{figure}
\centering 
    \includegraphics[width=3in]{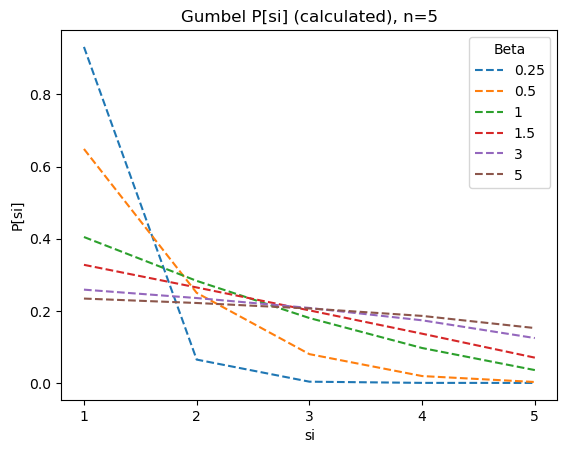}
    \includegraphics[width=3in]{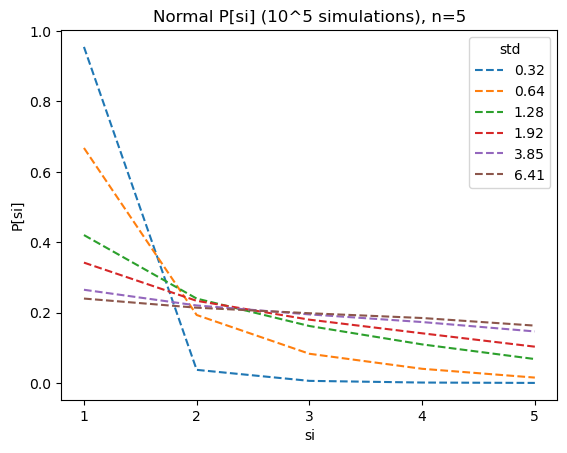}
    \caption{Figure of $P[\perm^i]$ for RUM with Gumbel and Gaussian noise, respectively. Gumbel is calculated exactly while Gaussian is simulated numerically.}\label{fig:psis}
\end{figure}

\begin{figure}
\centering 
    \includegraphics[width=3in]{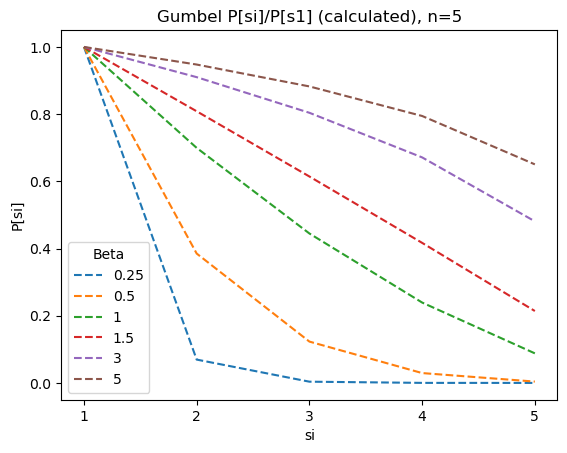}
    \includegraphics[width=3in]{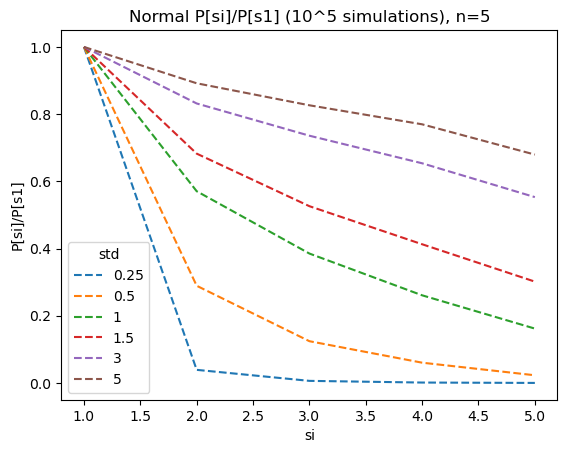}
    \caption{Figure of $\frac{P[\perm^1]}{P[\perm^i]}$ for RUM with Gumbel and Gaussian noise, respectively. Gumbel is calculated exactly while Gaussian is simulated numerically. Note that $P[\sigma^1]/P[\sigma^i]$ is strictly increasing as noise decreases for both Gumbel and Gaussian noise.}\label{fig:ps1psis}
\end{figure}

\subsection{General proofs for Section \ref{sec:optstrat}}\label{app:superstarstratproof}

\alwaysborf*
\begin{proof}
In order to prove this, we will prove that the expected value of items given the status vector $\avail$ is descending, or: 
$$\mathbb{E}_{\perm \sim \permdist}[\val_{\perm_i} \mid \avail] \geq \mathbb{E}_{\perm \sim \permdist}[\val_{\perm_i} \mid \avail] 
 \quad  i <j$$
 If this condition holds, then for any two $i, j$ given $\avail_i = \avail_j$, we know that the lower ranked of the two has higher utility, and thus no optimal strategy could every pick the higher ranked item. 

 Next, we will show that the superstar setting always satisfies this property. 
The expected value of the $i$th entry is given by: 
$$\mathbb{E}[\val_{\perm_i} \mid \avail] = \sum_{\perm \sim \permdist}P[\perm \mid \avail] \cd \val_{\perm_i} = \sum_{k = 1}^{\nitem} \frac{P[\avail \mid \perm^k ] \cd P[\perm^k]}{P[\avail]} \cd \val_{\perm^k_i}$$
where we have used that in the superstar setting, each of the permutations can be identified by the index of the high-value item. 
By identical reasoning, the expected value of the $j$th entry is given by: 
$$\mathbb{E}[\val_{\perm_j} \mid \avail] = \sum_{\perm \sim \permdist}P[\perm \mid \avail] \cd \val_{\perm_j} = \sum_{k = 1}^{\nitem} \frac{P[\avail \mid \perm^k ] \cd P[\perm^k]}{P[\avail]} \cd \val_{\perm^k_j}$$
The condition we wish to show is: 
$$\sum_{k = 1}^{\nitem} \frac{P[\avail \mid \perm_k ] \cd P[\perm_k]}{P[\avail]} \cd \val_{\perm^k_i} \geq \sum_{k = 1}^{\nitem} \frac{P[\avail \mid \perm_k ] \cd P[\perm_k]}{P[\avail]} \cd \val_{\perm^k_j}$$
Or: 
$$\sum_{k = 1}^{\nitem} P[\avail \mid \perm_k ] \cd P[\perm_k] \cd \p{\val_{\perm^k_i}- \val_{\perm^k_j}}\geq 0 $$
Because we are in the superstar setting, we know that $\val_{\perm^k_i}= \val_{\perm^k_j}$ (both items are low value) unless $k \in \{i, j\}$. Dropping all but these values of $k$ and using $\val_1 \geq \val_2$ for the high and low values respectively gives:
$$P[\avail \mid \perm^i] \cd P[\perm^i] \cd (\val_1 - \val_2) + P[\avail \mid \perm^j] \cd P[\perm^j] \cd (\val_2 - \val_1) \geq 0$$
which is positive whenever: 
$$(\val_1 - \val_2) \cd (P[\avail \mid \perm^i] \cd P[\perm^i] - P[\avail \mid \perm^j] \cd P[\perm^j]) \geq 0$$
By assumption, $\val_1 \geq \val_2$, and $\avail_i = \avail_j$. This latter assumption tells us that $P[\avail \mid \perm^i] = P[\avail \mid \perm^j]$: because the status in index $i, j$ are identical, they are equally likely to occur if the high-value item is in index $i$ or $j$. Finally, we require that the expected value of the distribution $\mathbb{E}{\perm \sim \permdist}[\val_{\perm_i}]$ is always descending in $i$, which for the superstar setting implies that the high-value item is always (weakly) more likely to be at lower indices ($P[\perm^i] \geq P[\perm^j]$). Taken together, this implies that the above condition reduces to:
$$(\val_1 - \val_2) \cd P[\avail \mid \perm^i]  \cd ( P[\perm^i] - P[\perm^j]) \geq 0$$
which always holds. 
\end{proof}

\counter*
\begin{proof}
Consider a setting where $\nitem=3$ and we have a completely uninformative ranking: $P[\perm^1] = P[\perm^2] = P[\perm^3]$. Then, our posterior will be completely unrelated to the ranking and depend solely on the status vector. 

We will consider our posterior after two different status vectors, and shown that the optimal strategy described by both of these shows that neither first-free nor first-busy can be optimal, because there exist settings where the firm would preferentially pick either free or busy candidates. 

\noindent \textbf{Case 1: [0, 0, 1]}\\
From our analysis of Theorem \ref{thrm:vposntwoaccstrat}, picking the bottom ranked (free) candidate is optimal when Equation \ref{eq:firstbusyeq} is satisfied: 
$$P[\perm^1] \cd \ratio \cd (\val_1/\busypen_1 - \val_2) \leq P[\perm^j] \cd (\val_1 -  \val_2/\busypen_2)  + \ratio \cd \val_2 \cd (1-1/\busypen_2) \cd P[\perm^{2\leq i < j}]+ \val_2 \cd (1-1/\busypen_2) \cd \sum_{i=j+1}^{\nitem} \ratio^{1-\avail_i} \cd P[\perm^i]$$
which in our setting with $\nitem=3$ and $P[\perm^1] = P[\perm^2] = P[\perm^3]$ simplifies to: 
$$\ratio \cd (\val_1/\busypen_1 - \val_2) \leq (\val_1 - \val_2/\busypen_2)  + \ratio \cd \val_2 \cd (1-1/\busypen_2)$$
Simplifying gives: 
$$\ratio \cd (\val_1/\busypen_1  - \val_2 - \val_2 \cd (1-1/\busypen_2))\leq \val_1 - \val_2/\busypen_2 $$
$$\ratio \cd (\val_1/\busypen_1 - \val_2 \cd (2-1/\busypen_2))\leq \val_1 -  \val_2/\busypen_2$$
Given $\val_1/\busypen_1  - \val_2 \cd (2-1/\busypen_2)>0$, this implies: 
$$\ratio \leq \frac{\val_1 - \val_2/\busypen_2 }{\val_1 /\busypen_1 - \val_2 \cd (2-1/\busypen_2)}$$

\noindent \textbf{Case 2: [1, 1, 0]}\\
From our analysis of Theorem \ref{thrm:vposntwoaccstrat}, picking the last ranked (busy) candidate is optimal when Equation \ref{eq:firstfreeeq} is satisfied: 

$$P[\perm^1] \cd (\val_1 - \val_2/\busypen_2) + \val_2 \cd (1-1/\busypen_2) \cd P[\perm^{2\leq i<j}] + (1-1/\busypen_2) \cd \val_2 \cd \sum_{i=j+1}^{\nitem} \ratio^{1-\avail_i} \cd P[\perm^i] \leq P[\perm^j] \cd \ratio \cd ( \val_1/\busypen_1 - \val_2)$$
which in our setting with $\nitem=3$ and $P[\perm^1] = P[\perm^2] = P[\perm^3]$ simplifies to: 

$$(\val_1 - \val_2/\busypen_2) + \val_2 \cd (1-1/\busypen_2) \leq \ratio \cd (1/\busypen_1 \cd \val_1 - \val_2)$$
Simplifying gives: 
$$\val_1 + \val_2 \cd (1-2/\busypen_2) \leq \ratio \cd (\val_1/ \busypen_1 - \val_2)$$
Assuming $\val_1/ \busypen_1 - \val_2>0$, this implies: 
$$\frac{\val_1 + \val_2 \cd (1-2 /\busypen_2)}{\val_1 /\busypen_1 - \val_2}\leq \ratio $$

\noindent \textbf{Analysis:}
Taken together, Case 1 and Case 2 imply that the optimal strategy given observed status vectors $[0, 0, 1]$ and $[1, 1, 0]$ is always to pick the bottom candidate when the condition below holds: 
$$\frac{\val_1 + \val_2 \cd (1-2/\busypen_2)}{\val_1 \cd/\busypen_1 - \val_2} < \ratio < \frac{\val_1 -  \val_2/\busypen_2}{\val_1/\busypen_1  - \val_2 \cd (2-1/\busypen_2)} $$
Note that if $\val_2=0$ this evaluates to: 
$$\busypen_1 < \ratio < \busypen_1$$
which is never satisfied. However, with $\val_2>0$, it is possible to set parameters such that this evaluates to a positive term: consider $\val_1 = 1, \val_2 = 0.3, \busypen_1 = \busypen_2 = 2$, where the bounds evaluate to: 
$$7.25 \leq \ratio \leq 17$$
In this setting, we have shown that the optimal strategy can \emph{never} be approximated by \enquote{first-free, first-busy}. However, this is also a setting where the ranking tool is extremely uninformative, and so we shouldn't have expected to have this be a reasonable strategy. Effectively, this shows that if $\val_2>0$, we can only hope for approximate optimality of strategies like this. 
\end{proof}

\vposntwoaccstrat*
\begin{proof}
We will begin with the setting where $\avail_{i<j} = 0, \avail_j= 1$ (the \textbf{highest-ranked candidate is busy}). We pick the first candidate whenever: 
$$\mathbb{E}_{\perm \sim \permdist}[\val_{\perm_1} \mid \avail]/ \gamma \geq  \mathbb{E}_{\perm \sim \permdist}[\val_{\perm_j} \mid \avail] $$
Using similar analysis to the proof of Lemma \ref{lem:alwaysborf}, we can apply Bayes rule to rewrite this as: 
$$\sum_{\ell = 1}^{\nitem} \frac{P[\avail \mid \perm^\ell ] \cd P[\perm^\ell]}{P[\avail]} \cd \val_{\perm^\ell_1}/\gamma_{\perm^\ell_1} \geq \sum_{\ell = 1}^{\nitem} \frac{P[\avail \mid \perm^\ell ] \cd P[\perm^\ell]}{P[\avail]} \cd \val_{\perm^\ell_j} $$
$$\sum_{\ell = 1}^{\nitem} P[\avail \mid \perm^\ell ] \cd P[\perm^\ell]  \cd \val_{\perm^\ell_1} /\gamma_{\perm^\ell_1} \geq \sum_{\ell = 1}^{\nitem} P[\avail \mid \perm^\ell ] \cd P[\perm^\ell]\cd \val_{\perm^\ell_j} $$
Because we are in the superstar setting, we know that $\val_{\perm_i^\ell} = \val_2$ unless $i=\ell$, so we can rewrite the condition as: 
$$P[\avail \mid \perm^1]/\busypen_1  \cd P[\perm^1] \cd \val_1 +  1/\busypen_2\sum_{i=2}^{\nitem} P[\avail \mid \perm^i] \cd P[\perm^i] \cd \val_2 \geq P[\avail \mid \perm^j] \cd P[\perm^1] \cd \val_1 + \sum_{i\ne j} P[\avail \mid \perm^i] \cd P[\perm^i] \cd \val_2$$
Again using similar analysis to Lemma \ref{lem:alwaysborf}, we can reason about $P[\avail \mid \perm^i]$, for a general $i$. Specifically, we can show that: 

$$
P[\avail \mid \perm^i] = P[\avail \mid \perm^j]  \quad \avail_i = \avail_j = 1$$
$$P[\avail \mid \perm^i] = \ratio \cd P[\avail \mid \perm^j]  \quad \avail_i= 0 \ne \avail_j = 1$$
where the probability of observing $\avail$ is \emph{higher} if the high value candidate is in an index $i$ with $\avail_i=0$, because the high value candidate has $\pfree_1 < \pfree_2$ and is thus less likely to be free. Using this and the fact that $\avail_{i<j}=0$, we can rewrite the conditions as:
$$P[\avail \mid \perm^j] \p{\ratio \cd P[\perm^{1}] \cd \val_1/\busypen_1 +  \ratio \cd P[\perm^{2\leq i<j}] \cd \val_2/\busypen_2  + P[\perm^j] \cd \val_2/\busypen_2  +  \sum_{i=j+1}^{\nitem}  \ratio^{1-\avail_i} \cd P[\perm^i] \cd \val_2/\busypen_2}$$
$$ \geq P[\avail \mid \perm^j] \p{P[\perm^{j}] \cd \val_1 +  \ratio \cd P[\perm^{1\leq i<j}] \cd \val_2 + \sum_{i=j+1}^{\nitem}  \ratio^{1-\avail_i} \cd P[\perm^i] \cd \val_2}$$
Note the common $P[\avail \mid \perm^j]$ term we can drop. Next, we can strategically rearrange and collect common terms. 
\begin{equation}\label{eq:firstbusyeq}
P[\perm^1] \cd \ratio \cd (\busypen_1 \cd \val_1 - \val_2) \geq P[\perm^j] \cd (\val_1 - \val_2/\busypen_2 )  + \ratio \cd \val_2 \cd (1-1/\busypen_2) \cd P[\perm^{2\leq i < j}]+ \val_2 \cd (1-1/\busypen_2) \cd \sum_{i=j+1}^{\nitem} \ratio^{1-\avail_i} \cd P[\perm^i]
\end{equation}
Next, we will apply the condition in the statement of this theorem, which tells us when we will pick the first candidate. Note that this depends on whether $R$ is less than or greater than 1. If $R<1$, the firm is using the first-free strategy, and will pick the $j$th candidate if and only if: 
$$\frac{P[\sigma^1]}{P[\sigma^j]} \leq \frac{1}{R}  = \frac{\val_1 - \val_2 /\busypen_2}{(\val_1/\busypen_1 - \val_2) \cd \ratio}$$
which can be rewritten as: 
\begin{equation}\label{eq:firstfreestrateq}
P[\perm^1] \cd (\val_1 / \busypen_1 - \val_2) \cd \ratio \leq P[\perm^j] \cd \p{\val_1 - \val_2 / \busypen_2}
\end{equation}
Conversely, if $R>1$, the firm is using the first-busy strategy, and will pick the $j$th candidate if and only if: 
$$\frac{P[\sigma^1]}{P[\sigma^j]} \leq R  = \frac{(\val_1 /\busypen_1 - \val_2) \cd \ratio}{\val_1 - \val_2 / \busypen_2}$$
or rewritten, whenever: 
\begin{equation}\label{eq:firstbusystrateq}P[\perm^1] \cd (\val_1 - \val_2 /\busypen_2) \leq P[\perm^j] \cd (\val_1/ \busypen_1 - \val_2) \cd \ratio
\end{equation}
First, we note that if a firm is using first-free and Equation \ref{eq:firstfreestrateq} is satisfied, then Equation \ref{eq:firstbusyeq} \emph{never} satisfied, and so it is always optimal to pick the $j$th (free) candidate.

Conversely, if Equation \ref{eq:firstfreestrateq} is \emph{not} satisfied, then Equation \ref{eq:firstbusyeq} is satisfied with error up to: 
$$\frac{P[\avail \mid \perm^j]}{P[\avail]} \p{\ratio \cd \val_2 \cd (1-1/\busypen_2) \cd P[\perm^{2\leq i < j}]+ \val_2 \cd (1-1/\busypen_2) \cd \sum_{i=j+1}^{\nitem} \ratio^{1-\avail_i} \cd P[\perm^i]}$$
$$ \leq \frac{p_2}{p_1} \cd \ratio \cd \val_2 \cd (1-1/\busypen_2) \cd (1-P[\perm^1] - P[\perm^j]) $$
$$ = \frac{p_2}{p_1} \cd \frac{p_2/(1-p_2)}{p_1/(1-p_1)}\cd \val_2 \cd (1-1/\busypen_2) \cd (1-P[\perm^1] - P[\perm^j])$$
$$ = \frac{p_2}{p_1}  \cd (1-p_2)^{\nitem-2} \cd \val_2 \cd (1-1/\busypen_2) \cd (1-P[\perm^1] - P[\perm^j])$$
$$ \leq 0.25 \cd \frac{p_2}{p_1} \cd \val_2 \cd (1-1/\busypen_2) \cd (1-P[\perm^1] - P[\perm^j]) $$
where in the last step we have used that $x \cd (1-x) \leq 0.25$ and $1-p_1 < 1$, and in the first step we have used that: 
$$\frac{P[\avail \mid \perm^j]}{P[\avail]} = \frac{P[\avail \mid \perm^j]}{P[\avail \mid \perm^j]}{\sum_{i=1}^NP[\avail \mid \perm^i] \cd P[\perm^i]} \leq \max_{i \in [N]} \frac{P[\avail \mid \perm^j]}{P[\avail \mid \perm^i]} = \frac{p_2}{p_1}$$
where the last step comes from the fact that $\perm_i, \perm_j$ differ only in a single index.\\ 
Next, we note that if a firm is using first-busy and the top candidate is busy, the firm will always pick it. Specifically, this means that Equation \ref{eq:firstbusystrateq} is \emph{not} satisfied, which means: 
$$P[\perm^1] \cd (\val_1 - \val_2 / \busypen_2) \geq P[\perm^j] \cd (\val_1/ \busypen_1 - \val_2) \cd \ratio \quad \Rightarrow \quad P[\perm^1] \cd \ratio \cd (\val_1/\busypen_1 - \val_2) \geq P[\perm^j] \cd (\val_1 - 1/\busypen_2 \cd \val_2)$$
because in the setting where a firm is using first-busy, we know that $R>1$ or $\ratio \cd (\val_1 /\busypen_1 - \val_2) \geq \val_1 - \val_2/ \busypen_2$. Thus, Equation \ref{eq:firstbusyeq} is satisfied up to error $\epsilon^*$, by similar reasoning as above. 
\\

Next, we move to the setting where $\avail_{i<j} =1, \avail_j = 0$ (the \textbf{highest-ranking candidate is free}). We pick the first candidate whenever: 
$$\mathbb{E}_{\perm \sim \permdist}[\val_{\perm_1} \mid \avail] \geq  \mathbb{E}_{\perm \sim \permdist}[\val_{\perm_j} \mid \avail]/ \busypen$$
By our prior analysis, we know this can be rewritten as:
$$P[\avail \mid \perm^1] \cd P[\perm^1] \cd \val_1 + \sum_{i=2}^{\nitem} P[\avail \mid \perm^i] \cd P[\perm^i] \cd \val_2 \geq  P[\avail \mid \perm^j] \cd P[\perm^j] \cd \val_1/\busypen_1  + \sum_{i\ne j} P[\avail \mid \perm^i] \cd P[\perm^i] \cd \val_2/\busypen_2 $$
Next, we again reason about $P[\avail \mid \perm^i]$ for an arbitrary $i$. Note that the $\avail$ status vector is different than in the previous part of this proof, so we will have a different correspondence: 
$$
P[\avail \mid \perm^i] = P[\avail \mid \perm^1]  \quad \avail_i = \avail_1 = 1$$
$$P[\avail \mid \perm^i] = \ratio \cd P[\avail \mid \perm^1]  \quad \avail_i= 0 \ne \avail_1 = 1$$
This enables us to rewrite the condition as: 
$$P[\avail \mid \perm^1]  \cd \p{\val_1 \cd P[\perm^1] + P[\perm^{2\leq i <j}] \cd \val_2 + \ratio \cd P[\perm^j] \cd \val_2 + \sum_{i=j+1}^{\nitem}  \ratio^{1-\avail_i} \cd P[\perm^i] \cd \val_2} $$
$$\geq P[\avail \mid \perm^1] \cd \p{\busypen_1 \cd \ratio \cd P[\perm^j] \cd \val_1 + \val_2/ \busypen_2 \cd P[\perm^{1 \leq i <j}] + \sum_{i=j+1}^{\nitem}  \ratio^{1-\avail_i} \cd P[\perm^i] \cd \val_2/\busypen_2} $$
Dropping a common multiplicative term and strategically grouping related terms gives: 
\begin{equation}\label{eq:firstfreeeq}
P[\perm^1] \cd (\val_1 - 1/\busypen_2 \cd \val_2) + \val_2 \cd (1-1/\busypen_2) \cd P[\perm^{2\leq i<j}] + (1-1/\busypen_2) \cd \val_2 \cd \sum_{i=j+1}^{\nitem} \ratio^{1-\avail_i} \cd P[\perm^i] \geq P[\perm^j] \cd \ratio \cd (\val_1/\busypen_1  - \val_2)
\end{equation}
We will again apply the conditions within Equations \ref{eq:firstfreestrateq}, \ref{eq:firstbusystrateq} to give us guarantees for which candidates we should pick. First, we note that if a firm is using first-free and the top-ranked candidate is free, the firm will always pick them. In particular, this means that Equation \ref{eq:firstfreestrateq} is satisfied, then 
$$P[\perm^1] \cd (\val_1/ \busypen_1 - \val_2) \cd \ratio \leq P[\perm^j] \cd (\val_1 - \val_2/ \busypen_2) \quad \Rightarrow \quad P[\perm^j] \cd (\val_1/ \busypen_1 - \val_2) \cd \ratio \leq P[\perm^1] \cd (\val_1 - \val_2/ \busypen_2)$$
because $P[\perm^1] \geq P[\perm^j]$, which implies that Equation \ref{eq:firstfreeeq} is satisfied and it is always optimal to pick the first (free) candidate. \\
Next, we will turn to the case where a firm is using first-busy. If Equation \ref{eq:firstbusystrateq} is satisfied, then Equation \ref{eq:firstfreeeq} is satisfied up to error:
$$\leq \frac{P[\avail \mid \perm^1]}{P[\avail]} \cd \p{\val_2 \cd (1-1/\busypen_2) \cd P[\perm^{2\leq i<j}] + (1-1/\busypen_2) \cd \val_2 \cd \sum_{i=j+1}^{\nitem} \ratio^{1-\avail_i} \cd P[\perm^i]}$$
$$\leq \frac{p_2}{p_1} \cd \p{\val_2 \cd (1-1/\busypen_2) \cd P[\perm^{2\leq i<j}] + (1-1/\busypen_2) \cd \val_2 \cd \ratio\cd P[\perm^{i>j}]} $$
$$< \frac{p_2}{p_1} \cd \val_2 \cd (1-1/\busypen_2) \cd (1-P[\perm^1] - P[\perm^j]) = \epsilon^*$$
as desired. Finally, we consider the case where the firm is using first-busy, but Equation \ref{eq:firstbusystrateq} is \emph{not} satisfied, or: 
$$P[\perm^1] \cd (\val_1 - \val_2/ \busypen_2) \geq P[\perm^j] \cd (\val_1 /\busypen_1 - \val_2) \cd \ratio$$
which we note implies that Equation \ref{eq:firstfreeeq} is always satisfied, and thus has 0 error. 
\end{proof}

\section{Proofs for Section \ref{sec:welfare}}\label{app:welfare}

Note: all proofs in this section only rely on accuracy properties as stated in Definition \ref{def:accuracy}: 
\accuracy*

\welfarefirm*
\begin{proof}
First, Lemma \ref{lem:convertacc} shows that it is always possible to take some status vector $\avail$ generated by a more accurate permutation distribution $\{P[\rho^i]\}$ and convert it into a status vector $\avail'$ generated by a less accurate permutation distribution $\{P[\perm^i]\}$ by iteratively swapping pairs of candidates. We will also use Lemma \ref{lem:alwaysborf}, which showed that the top-ranked busy candidate always has higher expected value than any other busy candidate (and similarly the top-ranked free candidate always has higher expected value than any other free candidate). 

Consider the candidate picked from perturbed status vector $\avail'$ using the optimal strategy for the less-accurate distribution $\{P[\perm^i]\}$. If it is the same as the candidate picked from the original status vector $\avail$ with the more accurate ranking tool $\{P[\rho^i]\}$, then the utility is exactly the same as with the less accurate ranking. 

However, if a different candidate is picked under $\avail'$, then we know that this candidate corresponds to an element in $\avail$ lower expected value than the candidate picked under $\avail$ with the more accurate distribution $\{P[\rho^i]\}$: as a result, the firm responding optimally with less accurate distribution $\{P[\rho^i]\}$ only has lower utility, as desired. 
\end{proof}

\begin{lemma}\label{lem:convertacc}
Given any ranking distributions $\{P[\rho^i]\}$, $\{P[\perm^i]\}$ such that $\{P[\rho^i]\}$ is more accurate according to Definition \ref{def:accuracy} and a status vector $\avail$ generated by $\{P[\rho^i]\}$, it is always possible to permute $\avail$ such that it is equivalent to a sample from the less-accurate distribution $\{P[\perm^i]\}$. 
\end{lemma}
\begin{proof}
First, by Definition \ref{def:accuracy}, we know that increasing accuracy of a permutation causes majorization, or $\sum_{i=1}^jP[\rho^i] \geq \sum_{i=1}^jP[\perm^i]$ for all $j$. We can view this as saying that the cumulative density function for the more accuracy permutation is always higher: we will use this to guide how we transform the more accurate distribution into the less accurate one. 

First, we know that $\sum_{i=1}^{\nitem}P[\rho^i] = \sum_{i=1}^{\nitem}P[\perm^i] =1$ (the total probability sums up to 1). Consider the smallest index $k$ such that $\sum_{i=1}^{k}P[\rho^i] > \sum_{i=1}^{k}P[\perm^i]$: this is the smallest index such that the more accurate distribution has a greater CDF than the less accurate distribution. This implies that the probability of the high value candidate being in index $k+1$ is larger in the less accurate distribution: $P[\rho^{k+1}] < P[\perm^{k+1}]$. 

Similarly, define $\ell$ to be the smallest index such that the CDF is strictly higher with the more accurate distribution: $P[\rho^{\ell}] > P[\perm^{\ell}]$. 

Then, we can flip indices $\ell$ and $k+1$ with a fixed probability so as to create a new distribution with $P[\rho^{k+1}]' = P[\perm^{k+1}]$ and $P[\rho^{\ell}]' \geq P[\perm^{\ell}]$. Specifically, we define $\alpha_{k+1}$ such that:
$$P[\perm^{k+1}] = P[\rho^{k+1}] \cd (1-\alpha_{k+1}) + P[\rho^{\ell}] \cd \alpha_{k+1}$$
Note that flipping indices $\ell, k+1$ also induces a change in $P[\rho^{\ell}]$: if doing so would cause $P[\rho^{\ell}] < P[\perm^{\ell}]$, we instead set $\alpha_{k+1}$ to be the value such that $P[\rho^{\ell}]'= P[\perm^{\ell}]$: 
$$P[\perm^{\ell}] = P[\rho^{\ell}] \cd (1-\alpha_{k+1}) + P[\rho^{k+1}] \cd \alpha_{k+1}$$ 
If necessary, we then repeat the process with the new minimum index $\ell'$ such that $P[\rho^{\ell'}] > P[\perm^{\ell'}]$ until we arrive at $P[\rho^{k+1}]' = P[\perm^{k+1}]$. 

After this is completed, we move on to the next $j<k$ such that $\sum_{i=1}^{j}P[\rho^i] > \sum_{i=1}^{j}P[\perm^i]$, swapping permutations until we have $P[\rho^i]' = P[\perm^i]'$.
\end{proof}

\increaseddupnew*
\begin{proof}
First, we will consider the case where a firm is using the \enquote{first-free} strategy, for some window $k\geq 1$. \\
If this is the case, that the probability that a firm picks a busy candidate is given by the probability that every candidate in the top $k$ is busy. This occurs with probability: 
$$P[\text{top $k$  busy} \mid \text{High value in top $k$}] \cd P[ \text{High value in top $k$}] $$
$$+ P[\text{top $k$  are busy} \mid \text{High value NOT in top $k$}] \cd P[ \text{High value NOT in top $k$}] $$
$$ = P[\perm^{1\leq i\leq k}] \cd (1-p_1) \cd (1-p_2)^{k-1} + (1-P[\perm^{1\leq i\leq k}]) \cd (1-p_2)^{k} $$
$$ = (1-p_2)^{k-1} \cd \p{P[\perm^{1\leq i\leq k}] \cd (p_2-p_1) + (1-p_2)}$$
We know that $p_2>p_1$ and increased accuracy increases $P[\perm^{1\leq i\leq k}]$ (by Definition \ref{def:accuracy}). Thus, if the ranking becomes more accurate and the firm's strategy ($k$) is held fixed, the probability of the firm picking a busy candidate increases. \\

What happens if accuracy increases and the firm's strategy changes to some top $j<k$? We will establish that the probability of picking a busy candidate is always higher when the firm is using first-free with window size $j$ than window size $k$ for $j<k$ (holding accuracy constant).
What we want to show is: 
$$(1-p_2)^{k-1} \cd \p{P[\perm^{1\leq i\leq k}] \cd (p_2-p_1) + (1-p_2)} < (1-p_2)^{j-1} \cd \p{P[\perm^{1\leq i\leq j}] \cd (p_2-p_1) + (1-p_2)}$$
It suffices to show this for $j=k-1$ (because by induction we can show this for all $j<k$). 
$$(1-p_2)^{k-1} \cd \p{P[\perm^{1\leq i\leq k}] \cd (p_2-p_1) + (1-p_2)} < (1-p_2)^{k-2} \cd \p{P[\perm^{1\leq i\leq k-1}] \cd (p_2-p_1) + (1-p_2)}$$
$$(1-p_2) \cd \p{P[\perm^{1\leq i\leq k}] \cd (p_2-p_1) + (1-p_2)} <  P[\perm^{1\leq i\leq k-1}] \cd (p_2-p_1) + (1-p_2)$$
Pulling over terms gives: 
$$(p_2 - p_1) \cd \p{(1-p_2) \cd P[\perm^{1\leq i\leq k}]  - P[\perm^{1\leq i\leq k-1}]} < (1-p_2) \cd (1-1+p_2)$$
$$(p_2 - p_1) \cd \p{(1-p_2) \cd P[\perm^{1\leq i\leq k}]  - P[\perm^{1\leq i\leq k-1}]} < (1-p_2) \cd p_2$$
Note that if the term inside the parentheses is negative, then this is automatically satisfied because the righthand side is positive. If not, then we can upper bound the lefthand side by using $p_2 - p_1 < p_2$, which cancels on both sides. 
$$(1-p_2) \cd P[\perm^{1\leq i\leq k}]  - P[\perm^{1\leq i\leq k-1}] < 1-p_2$$
Note that $P[\perm^{1\leq i\leq k}] \leq 1$ and $P[\perm^{1\leq i\leq k-1}]\geq 0$, which suffices to prove the result. 

Next, we will consider the case where a firm is using the \enquote{k-busy} strategy, for some $k\geq 1$. \\
We will find it helpful to instead analyze the probability that the firm picks a free candidate (which is exactly the complement of when it picks a busy candidate). 
A firm using a $k$-busy strategy picks a free candidate whenever all candidates in the top $k$ are free: 
$$P[\perm^{1\leq i\leq k}] \cd p_1 \cd p_2^{k-1} + (1-P[\perm^{1\leq i\leq k})] \cd p_2^k$$
$$ =  p_2^{k-1} \cd \p{P[\perm^{1\leq i\leq k}] \cd (p_1 - p_2)  + p_2}$$
Note that if the accuracy increases while $k$ stays fixed, this probability decreases because $p_1 < p_2$: so if the accuracy increases, the probability of the firm hiring a busy candidate increases. \\
What happens if the accuracy increases and the firm's strategy changes? Similarly to the previous analysis, we will show that the probability of picking a free candidate is higher for top $k<k+1$, so as accuracy increases and the firm decreases $k$, it will instantaneously increase the probability of picking a free candidate (decreasing the probability of picking a busy candidate). \\
Specifically, this occurs when: 

$$ p_2^{k-1} \cd \p{P[\perm^{1\leq i\leq k}] \cd (p_1 - p_2)  + p_2} > p_2^{k} \cd \p{P[\perm^{1\leq i\leq k+1}] \cd (p_1 - p_2)  + p_2}$$

$$P[\perm^{1\leq i\leq k}] \cd (p_1 - p_2)  + p_2 > p_2 \cd \p{P[\perm^{1\leq i\leq k+1}] \cd (p_1 - p_2)  + p_2}$$
$$p_2 \cd (1-p_2) > (p_1 - p_2) \cd (p_2 \cd P[\perm^{1\leq i\leq k+1}] - P[\perm^{1\leq i\leq k}] )$$
Note that $p_1-p_2<0$, so if the term inside the parentheses is positive, then the entire righthand side is negative, and this quantity automatically holds. If the term inside the parentheses is negative, we can rearrange the righthand side: 
$$p_2 \cd (1-p_2) > (p_2 - p_1) \cd (P[\perm^{1\leq i\leq k}] - p_2 \cd P[\perm^{1\leq i\leq k+1}] )$$
We can note that in this setting, we have:
$$P[\perm^{1\leq i\leq k}] > p_2 \cd P[\perm^{1\leq i\leq k+1}]$$
even though $P[\perm^{1\leq i\leq k}] < P[\perm^{1\leq i\leq k+1}]$. We can use the fact that: 
$$P[\perm^{1\leq i\leq k}] - p_2 \cd P[\perm^{1\leq i\leq k+1}] \leq P[\perm^{1\leq i\leq k+1}] - p_2 \cd P[\perm^{1\leq i\leq k+1}] = P[\perm^{1\leq i\leq k+1}]  \cd (1-p_2)$$
Applying this to the righthand side of the original equation, we have: 
$$(p_2 - p_1) \cd (P[\perm^{1\leq i\leq k}] - p_2 \cd P[\perm^{1\leq i\leq k+1}] ) \leq (p_2-p_1) \cd P[\perm^{1\leq i\leq k+1}]  \cd (1-p_2)\leq p_2 \cd (1-p_2)$$
as desired. 
\end{proof}

\increasedacctopclicknew*
\begin{proof}
First, we note that if a firm is using \enquote{follow the ranking}, then it always picks the top candidate, and as accuracy increases it continues to always pick the top candidate. 

Next, we consider the case that a firm is using \textbf{first-free} with window size $k$. A firm picks the top candidate either when a) the first candidate is free, or b) all top $k$ candidates are busy. What is probability that this happens? 
$$P[\text{top candidate free}] + P[\text{all top $k$ busy}]$$
$$ = p_1 \cd P[\perm^1] + p_2 \cd (1-P[\perm^1]) + (1-p_1) \cd (1-p_2)^{k-1} \cd P[\perm^{1<i\leq k}] + (1-p_2)^{k} \cd (1-P[\perm^{1 < i \leq k}])$$
Collecting like terms gives us: 
$$p_2 + (1-p_2)^k - P[\perm^1] \cd (p_2 - p_1) + (1-p_2)^{k-1} \cd P[\perm^{1\leq i \leq k}] \cd (1-p_1  - 1 + p_2)$$
$$=p_2 + (1-p_2)^k - P[\perm^1] \cd (p_2 - p_1) + (1-p_2)^{k-1} \cd P[\perm^{1\leq i \leq k}] \cd (p_2-p_1)$$
$$ = p_2 + (1-p_2)^k + (p_2 - p_1) \cd ((1-p_2)^{k-1} \cd P[\perm^{1\leq i \leq k}] - P[\perm^1])$$
Note that this is decreasing if and only if a more accurate ranking results in \emph{decreasing} the following quantity: 
$$(1-p_2)^{k-1} \cd P[\perm^{1\leq i \leq k}] - P[\perm^1] = (1-p_2)^{k-1} \cd P[\perm^{2\leq i\leq k}] - (1-(1-p_2)^{k-1}) \cd P[\perm^1]$$ 
Recall that from Definition \ref{def:accuracy} increased accuracy increases $\sum_{i=1}^jP[\perm^i]$ for all $j$: thus, the above quantity could \emph{increase} or \emph{decrease} the above quantity. See Figure \ref{fig:increasedecreasetop} for an example where this could lead to increased chances of picking the top candidate. 

Next, we will show that if the firm looks at the top $j<k$ then it always has a higher chance of picking the top candidate: this should be intuitive because there is a higher chance that all of the top $j$ candidates are busy than all of the top $k$ candidates are busy. This occurs when: 
$$p_2 + (1-p_2)^k + (p_2 - p_1) \cd ((1-p_2)^{k-1} \cd P[\perm^{1\leq i \leq k}] - P[\perm^1])$$
$$< p_2 + (1-p_2)^{k-1} + (p_2 - p_1) \cd ((1-p_2)^{k-2} \cd P[\perm^{1\leq i \leq k-1}] - P[\perm^1])$$
$$(1-p_2)^{k-1} \cd (1-(1-p_2)) > (p_2 - p_1) \cd ((1-p_2)^{k-1} \cd P[\perm^{1\leq i\leq k}] - (1-p_2)^{k-2} \cd P[\perm^{1\leq i \leq k-1}])$$
$$(1-p_2) \cd p_2 > (p_2 - p_1) \cd (1-p_2) \cd P[\perm^{1\leq i\leq k}] - P[\perm^{1\leq i \leq k-1}])$$
This is satisfied because $p_2> p_2-1$ and 
$$(1-p_2) \geq (1-p_2) \cd P[\perm^{1\leq i\leq k}] \geq (1-p_2) \cd P[\perm^{1\leq i\leq k}] - P[\perm^{1\leq i \leq k-1}])$$

Next, we consider the case where a firm is using \textbf{first-busy} with parameter $j$. The firm picks the top candidate either when a) the top candidate is busy or b) all top $j$ candidates are free. What is the probability that this happens? 
$$P[\text{top is busy}] + P[\text{all top $j$ are free}]$$
$$(1-p_1) \cd P[\perm^1] + (1-p_2) \cd (1-P[\perm^1]) + p_1 \cd p_2^{k-1} \cd P[\perm^{1<i\leq k}] + p_2^{k} \cd (1-P[\perm^{1 \leq i \leq k}])$$
Collecting like terms: 
$$1-p_2 + p_2^k + P[\perm^1] \cd (1-p_1-1 + p_2)   + p_2^{k-1}P[\perm^{1 \leq i\leq k}] \cd (p_1 - p_2)$$
$$ = 1-p_2 + p_2^k + (p_2-p_1) \cd (P[\perm^1] -p_2^{k-1}\cd P[\perm^{1\leq i\leq k}]) $$
Note that this is decreasing (holding $k$ constant) if and only if increasing accuracy causes the below quantity to decrease: 
$$P[\perm^1]\cd (1 -p_2^{k-1}) - p_2^{k-1} \cd P[\perm^{2\leq i\leq k}]$$
See Figure \ref{fig:increasedecreasetop} for an example where this could lead to increased chances of picking the top candidate.

Finally, we will show that this quantity increases for $j>k$, or: 
$$1-p_2 + p_2^k + (p_2-p_1) \cd (P[\perm^1] -p_2^{k-1}\cd P[\perm^{1\leq i\leq k}])<1-p_2 + p_2^{k-1} + (p_2-p_1) \cd (P[\perm^1] -p_2^{k-2}\cd P[\perm^{1\leq i\leq k-2}])$$
$$p_2^k + (p_2-p_1) \cd (-p_2^{k-1}\cd P[\perm^{1\leq i\leq k}])<p_2^{k-1} + (p_2-p_1) \cd (-p_2^{k-2}\cd P[\perm^{1\leq i\leq k-2}])$$
$$(p_2-p_1) \cd (p_2^{k-2}\cd P[\perm^{1\leq i\leq k-2}] -p_2^{k-1}\cd P[\perm^{1\leq i\leq k}])<p_2^{k-1}(1-p_2)$$
$$(p_2-p_1) \cd (P[\perm^{1\leq i\leq k-2}] -p_2\cd P[\perm^{1\leq i\leq k}])<p_2\cd (1-p_2)$$
This is satisfied because $p_2 > p_2 - p_1$ and 
$$P[\perm^{1\leq i\leq k-2}] -p_2\cd P[\perm^{1\leq i\leq k}] \leq (P[\perm^{1\leq i\leq k-1}] -p_2\cd P[\perm^{1\leq i\leq k}])  = (P[\perm^{1\leq i\leq k-1}](1 -p_2) \leq 1-p_2$$
\end{proof}

\section{Proofs for Section \ref{sec:beyondsuperstar}}\label{app:beyondsuperstar}

\subsection{Inversion-monotone properties (optimality of picking first free or first busy)}
Recall that in Section \ref{sec:optstrat} we proved Lemma \ref{lem:alwaysborf}: 

\alwaysborf*

We also noted that the precondition \enquote{In the superstar setting} was necessary. Here, we will show why this is: specifically, Lemma \ref{lem:counterexample} shows that in the non-superstar setting, it may be optimal to pick a strategy other than \enquote{first-busy, first-free}: for example, to pick the second-ranked free item rather than the top-ranked free item. 

Our main theorem is given by: 
\beyondsuperstartopbest*
We will prove this theorem through a series of smaller lemmas. 
\counterexample*
\begin{proof}
We will set the 3 candidates to have values $[\val_1 = 1, \val_2, \val_3= 0]$. We will create a permutation distribution with nonzero support on exactly two permutations: 
$$\perm^1 = [\val_1, \val_2, \val_3] \quad \perm^2 = [\val_3, \val_2, \val_1]$$
where $P[\perm^1] = 1-\epsilon, P[\perm^2] = \epsilon$. Before the candidate status vector is taken into account, the expected utility of picking candidate $i \in [1, 2, 3]$ is given by:
$$\mathbb{E}_{\perm \sim \permdist}[\val_{\perm_1}] = \val_1 \cd (1-\epsilon) \quad \mathbb{E}_{\perm \sim \permdist}[\val_{\perm_2}] = \val_2 \quad \mathbb{E}_{\perm \sim \permdist}[\val_{\perm_3}] = \epsilon \cd \val_1$$
If we want the distribution to have descending expected value, we will set: 
$$\val_1 \cd (1-\epsilon) > \val_2 > \epsilon \cd v_1$$
We will next consider the probabilities of being free. For simplicity, we will set $\pfree_1 < \pfree_2 = \pfree_3$, so candidates 2 and 3 have the same probability of being free, which is greater than for candidate 1. 

Next, we will consider the case where we have realized status vector $[1, 1, 0]$. This increases our posterior belief that we are in $\perm^2$, which decreases the expected utility of picking the first candidate. 

We will formalize this next. The posterior value for the first candidate is given by: 
$$P[\val_{\perm_1}= \val_1 \mid \avail] \cd \val_1 = \frac{P[\avail \mid \perm^1] \cd P[\perm^1]}{P[\avail]} = \val_1 \cd \frac{\pfree_1 \cd \pfree_2 \cd (1-\pfree_2) \cd (1-\epsilon)}{P[\avail]}$$
The posterior value for the second candidate is always exactly $\val_2$, because both possible permutations involve having $\val_2$ ranked second. However, we will find it useful to strategically rewrite this probability : 
$$P[\val_{\perm_2} = \val_2 \vert \avail] \cd \val_2 = v_2  = \val_2 \cd \frac{P[\avail]}{P[\avail} = \val_2\cd \frac{P[\avail \mid \perm^1] \cd P[\perm^1] + P[\avail \mid \perm^2] \cd P[\perm^2]}{P[\avail]} $$
$$= \val_2 \cd \frac{\pfree_1 \cd \pfree_2 \cd (1-\pfree_2)\cd (1-\epsilon) + \pfree_2^2 \cd (1-\pfree_1) \cd \epsilon}{P[\avail]}$$
The posterior value for the third candidate is given by: 
$$P[\val_{\perm_3} = \val_1 \mid \avail] \cd \val_1 /\gamma =\val_1 \cd \frac{P[\avail \mid \perm^2] \cd P[\perm^2]}{P[\avail]} \gamma =  \val_1 \cd \frac{ \pfree_2^2 \cd (1-\pfree_1) \cd \epsilon}{P[\avail] \cd \gamma }$$
Next, we will show that there exists parameters such that the optimal strategy (candidate with highest posterior utility) is the second (free) candidate, rather than the first-free or first-busy strategy. Specifically, we will show that this condition holds for $\epsilon = 0.1, \val_1 = 1, \val_2 = 2/3, \val_3 = 0, \pfree_1 = 0.1, \pfree_2 = 0.4, \busypen \geq 1$. Note that these parameters immediately satisfies $v_2 < (1-\epsilon) \cd v_1$, which implies that the prior ranking is descending in expected utility (as desired).  

Note that each of these expected utility terms have a common denominator of $P[\avail]$, which we can drop. We know that picking the second (free) candidate has higher expected utility than picking the first (free) candidate whenever: 
$$\val_1 \cd \pfree_1 \cd \pfree_2 \cd (1-\pfree_2) \cd (1-\epsilon) < \val_2 \cd (\pfree_1 \cd \pfree_2 \cd (1-\pfree_2) \cd (1-\epsilon) + \pfree_2^2 \cd (1-\pfree_1) \cd \epsilon)$$
Dropping a common term of $\pfree_2$ and distributing $\val_2$: 
$$\val_1 \cd \pfree_1 \cd (1-\pfree_2) \cd (1-\epsilon) < \val_2 \cd \pfree_1 \cd (1-\pfree_2) \cd (1-\epsilon) + \val_2 \cd  \pfree_2 \cd (1-\val_1) \cd \epsilon$$
Substituting in for the given values sets: 
$$1 \cd 0.1 \cd 0.6 \cd 0.9 < \frac{2}{3} \cd 0.1 \cd 0.6 \cd 0.9 + \frac{2}{3} \cd 0.4 \cd 0.9 \cd 0.1$$
$$ 0.054 < 0.06$$
as desired. 

We will additionally require that picking the third candidate (given $\avail$ realized) has lower expected utility than picking the second candidate. This means we require: 
$$\val_2 \cd (\pfree_1 \cd \pfree_2 \cd (1-\pfree_2) \cd (1-\epsilon) + \pfree_2^2 \cd (1-\pfree_1) \cd \epsilon) > 1/\gamma \cd \val_1 \cd \pfree_2^2 \cd (1-\pfree_1) \cd \epsilon$$
$$\val_2 \cd (\pfree_1 \cd (1-\pfree_2) \cd (1-\epsilon) + \pfree_2 \cd (1-\pfree_1) \cd \epsilon) > 1/\gamma \cd \val_1 \cd \pfree_2 \cd (1-\pfree_1) \cd \epsilon$$
Substituting in for given values: 
$$v_2 \cd (p_1 \cd (1-p_2) \cd (1-\epsilon) + p_2 \cd (1-p_1) \cd \epsilon) > 1/\gamma \cd v_1 \cd p_2 \cd (1-p_1) \cd \epsilon$$
$$\frac{2}{3} \cd 0.1 \cd 0.6 \cd 0.9 + \frac{2}{3} \cd 0.4 \cd 0.9 \cd 0.1 >  1 \cd 0.4 \cd 0.9 \cd 0.1/\busypen$$
$$ 0.06 > 0.036/\busypen$$
which holds for any $\busypen \geq 1$. 
\end{proof}

Note that the ranking in Lemma \ref{lem:counterexample} did satisfy descending expected value. However, rankings were somewhat unusual - in particular, there were only two orderings with nonzero probability, even though, given 3 items, there are 6 possible permutations of each of them. A more natural ranking would probably have some nonzero weight over each of these, ideally with greater weight on rankings where more of the items are ordered \enquote{correctly} (that is, with $\val_{\perm_i} > \val_{\perm_j}$ for $i<j$). 

Definition \ref{def:monotone} exactly describes this intuition. Specifically, it defines an \enquote{inversion-monotone} probability distribution as one where, for any permutation $\perm$ with at least two items inverted (that is, $\val_{\perm_i} < \val_{\perm_j}$ for $i<j$), there exists another permutation $\tilde \perm$ that is identical to $\perm$ but with items in indices $i, j$ flipped, and with probability at least as high as $\perm$ ($P[\tilde \perm] \geq P[\perm]$). Lemma \ref{lem:monotoneenough} proves that, if a probability distribution is inversion-monotone, we regain the same property we had in the superstar setting: the optimal strategy for the firm will always be to pick the first free item or first busy item.

\monotonedef*

\monotoneenough*

While Lemma \ref{lem:monotoneenough} shows that inversion-monotonicity will guarantee a more straightforward strategy space, it is natural to wonder how reasonable such a requirement is. In fact, it turns out that multiple commonly-used models of permutations already satisfy this property. 

First, Lemma \ref{lem:mallowmonotone} shows that the Mallows model is inversion-monotone. While helpful, this is not extremely surprising: the Mallows model \cite{mallows1957non} is constructed such that the probability of a permutation $\perm$ occurring is directly tied to the number of pairs of items that are inverted, so this property follows naturally.  

\begin{restatable}{lemma}{mallowmonotone}
\label{lem:mallowmonotone}
The Mallows model is inversion-monotone. 
\end{restatable}

While the Mallows model is frequently used as a model of permutations, it does have drawbacks. Specifically, one desirable property of rankings is that items with very different true values should be less likely to be inverted. For example, given $[\val_1 = 10, \val_2 = 9, \val_3 = 0]$, we would expect items $\val_1, \val_2$ to be more frequently inverted than $\val_2, \val_3$, even though their ordinal ranks differ by the same amount. This type of property cannot be expressed by the Mallows model, but it can be captured by the Random Utility Model \cite{thurstone1994law}. As described in Section \ref{sec:modelassump}, in the Random Utility Model (RUM), while each item has some true value $\{\val_i\}$, it is assumed that they are ranked by noised versions of these values, such as additive noise ($\hat \val_i \sim \mathcal{D}$). This automatically satisfies the property that items with more similar true values will be more likely to be swapped. Additionally, it seems likely that the Random Utility Model might better capture the performance of rankings produced by humans or algorithmic tools, which might have some sense of the \enquote{true value} of each item, but make small, independent errors in their estimation of these values. However, does the Random Utility Model satisfy inversion-monotonicity? 

Theorem \ref{thrm:rummonotone} answers this question in the affirmative. While we believe that this property may be of independent interest, to our knowledge, we are the first paper to prove such a property for the Random Utility model. The proof of Theorem \ref{thrm:rummonotone} is largely proven by two sub-lemmas. First, Lemma \ref{lem:orderedgum} follows directly from properties of the Gumbel distribution. Next, Lemma \ref{lem:ordered} generalizes this result to all symmetric noise distribution (e.g., including Gaussian noise, a standard choice in the RUM): the proof of this lemma turns out to be surprisingly subtle.

\begin{restatable}{theorem}{rummonotone}
    \label{thrm:rummonotone}
The RUM with either a) Gumbel noise, or b) any noise from any identical, symmetric, noise distribution is inversion-monotone. 
\end{restatable}

\begin{restatable}{lemma}{orderedgum}
\label{lem:orderedgum}
Suppose we have two random variables given by $X_1 = \mu_1 + \epsilon_1, X_2 = \mu_2 + \epsilon_2$, for $\mu_1 > \mu_2$ and $\epsilon_1, \epsilon_2 \sim \mathcal{D}$ from a \emph{Gumbel distribution} $\mathcal{D}$. Then, if $\abs{X_1 - X_2} = \Delta$, 
$$P[X_1 > X_2 \mid \abs{X_1 - X_2} = \Delta] > P[X_1 < X_2 \mid \abs{X_1 - X_2} = \Delta]$$
\end{restatable}

\begin{restatable}{lemma}{ordered}
\label{lem:ordered}
Suppose we have two random variables given by $X_1 = \mu_1 + \epsilon_1, X_2 = \mu_2 + \epsilon_2$, for $\mu_1 > \mu_2$ and $\epsilon_1, \epsilon_2 \sim \mathcal{D}$ for a \emph{symmetric, single-peaked distribution} $\mathcal{D}$. Then, if $\abs{X_1 - X_2} = \Delta$, 
$$P[X_1 > X_2 \mid \abs{X_1 - X_2} = \Delta] > P[X_1 < X_2 \mid \abs{X_1 - X_2} = \Delta]$$
\end{restatable}

These results tell us that, even in settings beyond the superstar model, for cases where permutations are generated by commonly-used models, the optimal strategy will still be to either pick the first free item or the first busy candidate. 

Finally, we present proofs for all above lemmas: 

\monotoneenough*
\begin{proof}
The expected value of the $i$th entry is given by: 
$$\mathbb{E}[\val_{\perm_i} \mid \avail] = \sum_{\perm \in \permdist}P[\perm \mid \avail] \cd \val_{\perm_i}  = \sum_{\perm \in \permdist} \frac{P[\avail \mid \perm ] \cd P[\perm]}{P[\avail]} \cd \val_{\perm_i} $$
By identical reasoning, the expected value of the $j$th entry is given by: 
$$\mathbb{E}[\val_{\perm_j} \mid \avail] = \sum_{\perm \in \permdist}P[\perm \mid \avail] \cd \val_{\perm_j}  = \sum_{\perm \in \permdist} \frac{P[\avail \mid \perm ] \cd P[\perm]}{P[\avail]} \cd \val_{\perm_j} $$
We wish to show that: 
$$\sum_{\perm \in \permdist} \frac{P[\avail \mid \perm ] \cd P[\perm]}{P[\avail]} \cd \val_{\perm_i} \geq \sum_{\perm \in \permdist} \frac{P[\avail \mid \perm ] \cd P[\perm]}{P[\avail]} \cd \val_{\perm_j}$$

Consider some $\perm$ such that $\val_{\perm_i} < \val_{\perm_j}$. Then, we consider a unique mapping to some $\tilde \perm$ such that $\perm_k = \tilde \perm_k$ except for $\tilde \perm_i = \perm_j, \tilde \perm_j = \perm_i$ (the $i$th and $j$th elements are swapped).  

Then, we will show that: 
$$P[\perm \mid \avail] \cd \val_{\perm_i}  + P[\tilde \perm \mid \avail] \cd  \val_{\tilde \perm_i}  > P[\perm \mid \avail] \cd \val_{\perm_j}  + P[\tilde \perm \mid \avail] \cd  \val_{\tilde \perm_j} $$
By construction, we have $\tilde \perm_i = \perm_j, \tilde \perm_j = \perm_i$, so we can rewrite this as: 
$$P[\perm \mid \avail] \cd \val_{\perm_i}  + P[\tilde \perm \mid \avail] \cd  \val_{\perm_j}  > P[\perm \mid \avail] \cd \val_{\perm_j}  + P[\tilde \perm \mid \avail] \cd  \val_{\perm_i} $$
$$(P[\tilde \perm \mid \avail] - P[\perm \mid \avail]) \cd (\val_{\perm_j} - \val_{\perm_i}) > 0$$
Next, we can rewrite using Bayes rule: 
$$\p{ \frac{P[\avail \mid \tilde \perm] \cd P[\tilde \perm]}{P[\avail]}- \frac{P[\avail \mid \perm] \cd P[\perm]}{P[\avail]}} \cd (\val_{\perm_j} - \val_{\perm_i}) > 0 $$ 
Dropping the common denominator gives: 
$$\p{ P[\avail \mid \tilde \perm] \cd P[\tilde \perm] - P[\avail \mid \perm] \cd P[\perm]} \cd (\val_{\perm_j} - \val_{\perm_i})> 0 $$ 
We will argue that $P[\avail \mid \tilde \perm] = P[\avail \mid \perm] $. Recall that $\tilde \perm, \perm$ are identical except for at entries $i, j$. Because $\avail_i = \avail_j$ by assumption, for the probability of observing $\avail$ does not change if candidates are swapped between these two entries. Dropping this common term gives: 
$$\p{P[\tilde \perm] - P[\perm]}\cd (\val_{\perm_j} - \val_{\perm_i}) > 0 $$ 
The second term is satisfied because $\val_{\perm_j}> \val_{\perm_i}$ by assumption, and the first term holds because $P[\tilde \perm] > P[\perm]$ by the requirement that the distribution is inversion-monotone as in Definition \ref{def:monotone}. 

Finally, we will show that $\mathbb{E}[\val_{\sigma_i} \mid \avail] \geq \mathbb{E}[\val_{\sigma_j} \mid \avail]$ implies that the best strategy is to either pick the top ranked free or top-ranked busy candidate. \\
Consider any pairs of indices $i, j$ with $i<j$ and $\avail_i =\avail_j$ (both candidates free or busy). Then, the condition indicates that: 
$$\busypen^{1-\avail_i} \cd \mathbb{E}[\val_{\sigma_i} \mid \avail] \geq \busypen^{1-\avail_j} \cd \mathbb{E}[\val_{\sigma_j} \mid \avail]$$
Because there is a common term of $\busypen^{1-\avail_i}= \busypen^{1-\avail_j}$, this indicates that the firm only gets lower utility by picking the candidate in index $j$. As a result, the highest-ranked free candidate has the highest expected utility for the firm among candidates that are free, and the highest-ranked busy candidate has the highest expected utility for the firm among candidates that are busy. Thus, the optimal strategy is always to pick either the top ranked free or top ranked busy candidate. 
\end{proof}

\mallowmonotone*
\begin{proof}
Consider some permutation $\perm$, with $\val_{\perm_i} > \val_{\perm_j}$ for $i > j$. Recall that the probability of seeing a permutation is inversely proportional to the number of inversions (e.g. $k> l$ such that $\val_{\perm_k} < \val_{\perm_l}$). 

If $i = j+1$ then this is obvious: this only changes the number of inversions for $i, j$, and increases the number of inversions by exactly 1. 

If $i > j+1$, then there exists some set $S \in [i+1, j-1]$. Flipping $\val_{\perm_i}, \val_{\perm_j}$ also changes the number of inversions within this set (but maintains the number of inversions in elements outside of this set). We would like to show that if $\val_{\perm_i} > \val_{\perm_j}$, flipping the elements $i, j$ only increases the number of inversions (and thus makes this alternative permutation less likely). 

Within $S$, we can denote the set of elements that are \emph{greater} and \emph{less} than $i$ and $j$ respectively by:  
$$G_i = \{k \in S \mid \val_{\perm_k} > \val_{\perm_i}\} \quad G_j = \{k \in S \mid \val_{\perm_k} > \val_{\perm_j}\} $$
$$L_i = \{k \in S \mid \val_{\perm_k} < \val_{\perm_i}\} \quad L_j = \{k \in S \mid \val_{\perm_k} < \val_{\perm_j}\} $$
Note that $G_i \subseteq G_j$ and $L_j \subseteq$: the set of elements that is greater than $i$ is a subset of the set of elements that is greater than $j$, and the set of elements that is less than $j$ is a subset of the elements that is smaller than $i$. 

Note that within $S$ in the original permutation $\pi$, the number of inversions is given by 
$$\abs{G_i} + \abs{L_j}$$
the elements that are greater than $\val_{\perm_i}$ and less than $\val_{\perm_j}$. After swapping the order of elements $i, j$, the number of inversions within $S$ is given by: 
$$\abs{G_j} + \abs{L_i}$$
the elements that are greater than $\pi_j$ and less than $\pi_i$. By our prior reasoning: 
$$\abs{G_j} \geq \abs{G_i} \quad \abs{L_i} \geq \abs{L_j}$$
and so: 
$$\abs{G_j} + \abs{L_i} \geq \abs{G_i} + \abs{L_j}$$
We note that flipping $i, j$ involves at least one more inversion (since now we have $\val_{\perm_j} < \val_{\perm_i}$) and so the swapping process strictly increased the total number of inversions. 
\end{proof}

\rummonotone*
\begin{proof}
Consider any permutation $\perm$ with at least one inversion (an instance where $\val_{\perm_i}<\val_{\perm_j}$ but $i<j$). Such a permutation occurs when $\hat \val_{\perm_i} > \hat \val_{\perm_j}$.  By Lemmas \ref{lem:orderedgum} and \ref{lem:ordered}, we know that given $\val_{\perm_i}<\val_{\perm_j}$ and pair of noised values $\hat \val_i,  \hat \val_j$ generated with with noise either a) from a Gumbel distribution or b) any additive symmetric noise distribution, if $\abs{\hat \val_{\perm_i} - \hat \val_{\perm_j}} = \Delta$, then 
$$P[ \hat \val_{\perm_j}  > \hat \val_{\perm_i}\mid \abs{\hat \val_{\perm_i} - \hat \val_{\perm_j}}= \Delta] > P[ \hat \val_{\perm_j}  < \hat \val_{\perm_i} \mid \abs{\hat \val_{\perm_i} - \hat \val_{\perm_j}}= \Delta]$$
This tells us that it is more likely we would have $\hat \val_i < \hat \val_j$ (the correct relative ordering of these candidates) than their inversion. Because noise is added to element independently in the RUM, we know that the noised value $\hat \val_k$ is completely independent of the noised value of $\hat \val_i , \hat \val_j$ for $k \ne i, j$. Thus, we know that there must exist a permutation $\tilde \perm$ identical to $\perm$ except with the candidates in indices $i, j$ flipped, and we must have $P[\tilde \perm] \geq P[\perm]$.  
\end{proof}

\orderedgum*
\begin{proof}
First, note that the mean of a Gumbel distribution is given by $\mu + \beta \cd \eta$, where $\eta$ here is Euler's constant. That means that we can model $X_i = \mu_i + \epsilon_i$ as $X_i \sim G(X_i - \beta \cd \eta)$. Next, we consider the difference distribution induced by $\epsilon_1 - \epsilon_2$: we note that the difference distribution between two Gumbel distributions is given by a Logistic distribution governed by mean $X_1 - \beta \cd \eta - X_2 + \beta \cd \eta = X_1 - X_2$ and noise parameter $\beta$. Because $X_1 > X_2$, this distribution is shifted towards the right: thus, for all values $x$, the probability of $X_1 > X_2 $ given $\abs{X_1 - X_2} = x$ is always greater than the probability of $X_1 < X_2 $, as desired. 
\end{proof}

\ordered*
\begin{proof}
We will begin by analyzing the distribution $\mathcal{D}'$ induced by $X = X_1 - X_2$. 

\noindent \textbf{Symmetric:}\\
First, we will show that it is symmetric around $\mu_1 - \mu_2$. In order to show this, we will show that: 
$$P[X = \mu_1 - \mu_2 + \delta] = P[X = \mu_1 - \mu_2 - \delta]$$
for all $\delta >0$. Suppose that we have some $\epsilon_1, \epsilon_2$ such that: 
$$X = X_1 - X_2 = \mu_1 +\epsilon_1 - \mu_2 - \epsilon_2 = \mu_1 - \mu_2 + \delta$$
for $\delta = \epsilon_1 - \epsilon_3$. Then, we can show a mapping to an equally-likely event where $X = \mu_1 - \mu_2 - \delta$. Because $\epsilon_1, \epsilon_2$ are drawn from the same distribution $\mathcal{D}$, it is equally likely that we would have $\epsilon_1' = \epsilon_2, \epsilon_2' = \epsilon_1$. This would give us: 
$$X = X_1 - X_2 = \mu_1 + \epsilon_2 - \mu_2 - \epsilon_1 = \mu_1 - \mu_2 + (\epsilon_2 - \epsilon_1) = \mu_1 - \mu_2 - (\epsilon_1 - \epsilon_2) = \mu_1 - \mu_2 -  \delta $$
as desired. 

\noindent \textbf{Unimodal:}\\
 The next thing we would like to show is that $\mathcal{D}'$ is unimodal (strictly increasing and then decreasing).

We will look at the distribution $g(x)$ of $\mathcal{D}'$ directly. We'd like to show that this is decreasing for $x> \mu_1 - \mu_2$. What's the probability density of $g(\delta)$ for some $\delta  =\epsilon_1 - \epsilon_2$? We can obtain this by integrating over all possible values of $\epsilon_1, \epsilon_2$, along with the cdf $f(\cd)$ for the noise distribution: 
$$g(\delta) = \int_{-\infty}^{\infty}\mathbbm{1}[\delta = \epsilon_1 - \epsilon_2]f(\epsilon_1) \cd f(\epsilon_2) d\epsilon_1d\epsilon_2$$
We can rewrite as a single integral over $\epsilon_1$ and use $\epsilon_2 = \epsilon_1 - \delta$.
$$g(\delta) = \int_{-\infty}^{\infty}f(\epsilon_1)\cd f(\epsilon_1 - \delta)d\epsilon_1$$
We would like to show that this is decreasing in $\delta$ for $\delta>0$. We can take the derivative wrt $\delta$, which gives us: 
$$\frac{d}{d\delta}g(\delta) =  \frac{d}{d\delta}\int_{-\infty}^{\infty}f(\epsilon_1)\cd f(\epsilon_1 - \delta)d\epsilon_1$$
Integration and differentiation commutes, so we have: 
$$ =  \int_{-\infty}^{\infty}f(\epsilon_1)\cd \frac{d}{d\delta}f(\epsilon_1 - \delta)d\epsilon_1$$
by the chain rule: 
$$ = - \int_{-\infty}^{\infty}f(\epsilon_1)\cd f'(\epsilon_1 - \delta)d\epsilon_1$$
Which means we want to show that: 
\begin{equation}\label{eq:pos} \int_{-\infty}^{\infty}f(\epsilon_1)\cd f'(\epsilon_1 - \delta)d\epsilon_1 > 0\end{equation}
Let's consider a closely related term: 
$$ \int_{-\infty}^{\infty}f(\epsilon_1)\cd f'(\epsilon_1)d\epsilon_1$$
We know that $f(\cd)$ is symmetric around 0 and is increasing below 0 and decreasing above 0, which means that this above term equals 0. $f(\epsilon_1 - \delta)$ is shifted to the right: we will show that this suffices to show that Equation \ref{eq:pos} is positive. 

We can rewrite Equation \ref{eq:pos} as: 
$$\int_{-\infty}^{\infty}f(\epsilon+\delta) \cd f'(\epsilon)d\epsilon$$
And what we want to show is
$$\int_{-\infty}^{\infty}f(\epsilon+\delta) \cd f'(\epsilon)d\epsilon > \int_{-\infty}^{\infty}f(\epsilon)\cd f'(\epsilon_1)d\epsilon$$
Or rewritten: 
$$\int_{-\infty}^{\infty}\p{f(\epsilon+\delta) - f(\epsilon)}\cd f'(\epsilon)d\epsilon > 0$$
First, we can divide this into two components based on whether $\epsilon$ is positive or negative: 
$$\int_{-\infty}^0(f(\epsilon + \delta) - f(\epsilon)) \cd f'(\epsilon) d \epsilon + \int_{0}^{\infty}(f(\epsilon + \delta) - f(\epsilon)) \cd f'(\epsilon) d \epsilon$$
$$\int_{0}^{\infty}(f(-\epsilon + \delta) - f(-\epsilon)) \cd f'(-\epsilon) d \epsilon + \int_{0}^{\infty}(f(\epsilon + \delta) - f(\epsilon)) \cd f'(\epsilon) d \epsilon$$
Because $f$ is symmetric by our prior reasoning, we have that $f(\epsilon) = f(-\epsilon)$ and $f'(\epsilon) = - f'(-\epsilon)$, which allows us to rewrite: 
$$-\int_{0}^{\infty}(f(-\epsilon + \delta) - f(\epsilon)) \cd f'(\epsilon) d \epsilon + \int_{0}^{\infty}(f(\epsilon + \delta) - f(\epsilon)) \cd f'(\epsilon) d \epsilon$$
We can then combine these terms to give: 
$$\int_0^{\infty}f'(\epsilon) \cd (f(\epsilon + \delta) - f(\epsilon) - f(-\epsilon + \delta) + f(\epsilon))d\epsilon$$
$$=\int_0^{\infty}f'(\epsilon) \cd (f(\epsilon + \delta) - f(-\epsilon + \delta))d\epsilon$$
We know that $f'(\epsilon)$ is always negative for positive $\epsilon$ (by the assumption that $f(\cd)$ is unimodal and centered at 0). In order for the total term to be positive, we need to show that the other term is also negative - that is, that 
$$f(-\epsilon + \delta) > f(\epsilon + \delta) \quad  \forall \epsilon > 0, \delta > 0$$
Again, we use that $f(\cd)$ is symmetric and unimodal. This means that 
$$f(\epsilon + \delta) = f(-\epsilon - \delta) < f(-\epsilon + \delta)$$
as desired. 

\noindent\textbf{Overall reasoning: }
Having shown that $\mathcal{D}'$ is unimodal and symmetric, we will now show that $P[X = \Delta] > P[X= -\Delta]$ for all $\Delta >0$. We know that $\mathcal{D}'$ has a mean at $\mu_1 - \mu_2>0$. 

If $0 < \Delta < \mu_1 - \mu_2$, then the event $X = \Delta$ occurs on the lefthand side of the unimodal distribution, as does the event $X = -\Delta$. Because $\mathcal{D}'$ is unimodal and strictly decreasing, this implies that $P[X = -\Delta] < P[X = \Delta]$. 

If $0 < \mu_1 - \mu_2 \leq \Delta$, then $P[X = \Delta]$ occurs on the righthand side of the unimodal curve, at a distance $\Delta - (\mu_1 - \mu_2)$ above the peak. Then, reflecting across the axis of symmetry gives the point $P[X = (\mu_1 - \mu_2) - (\Delta - (\mu_1 - \mu_2))] = P[X = 2\cd (\mu_1 - \mu_2) - \Delta]  = P[X = \Delta]$. This point is on the lefthand side of the unimodal curve, but because $\mu_1 - \mu_2 > 0$, this point also is above the point $P[X = -\Delta]$, and therefore we have:
$$P[X = \Delta] = P[X = 2 \cd (\mu_1 - \mu_2) - \Delta] > P[X = -\Delta]$$
as desired. 
\end{proof}

\subsection{Beyond-the-superstar strategy}\label{app:beyondapp}

\begin{definition}[Accuracy: beyond superstar]\label{def:accuracybeyond}
A ranking distribution $\{P[\perm]'\}$ is \emph{more accurate} than $\{P[\perm]\}$ if both conditions hold: 

\begin{enumerate} 
     \item Ratio of probabilities: for every pair of permutations $\perm, \tilde \perm$ that differ solely by inverting $\val_i \geq \val_k$ in index $1, j>2$, the ratio of probabilities of finding $P[\perm]/P[\tilde \perm]$ a) increases with increased accuracy, and b) increases with increased index $j$. 
 \item Majorization: if the probability of a candidate of value $v_i$ being in the top $k$ for any $k\geq 1$ decreases, then the probability of any other candidate with value $v_j<v_i$ being in the top $k$ also decreases.
\end{enumerate}
\end{definition}

\begin{restatable}{lemma}{gumbelsataccuracy}\label{lem:gumbelsataccuracy}
Plackett-Luce (RUM with Gumbel noise) satisfies the accuracy definitions in Definition \ref{def:accuracybeyond}. 
\end{restatable}
\begin{proof}
\textbf{Ratio of probabilities:}
The quantity of interest is: 
$$\frac{P[\perm]}{P[\tilde \perm]} $$
where $\tilde \perm$ is identical to $\perm$, except with elements in $1, j$ flipped. By the definition of the Plackett-Luce model, this is given by: 

$$\prod_{j=1}^{\nitem}\frac{\exp(\val_{\perm_j}/\beta)}{\sum_{\ell=j}^{\nitem}\exp(\val_{\perm_{\ell}}/\beta)} \frac{\exp(\val_{\perm_j}/\beta)}{\prod_{j=1}^{\nitem}\sum_{\ell=j}^{\nitem}\exp(\val_{\perm_{\ell}}/\beta)} = \frac{\exp(\val_{\perm_j}/\beta)}{\prod_{j=1}^{\nitem}\sum_{\ell=j}^{\nitem}\exp(\val_{\perm_{\ell}}/\beta)}$$
We can immediately note two points of symmetry: 
\begin{itemize}
    \item The numerator is identical for $\perm, \tilde \perm$. 
    \item In the denominator, all elements in indices $\ell > j$ the elements within the product are identical. 
\end{itemize}
 The ratio of probabilities is given by: 

 $$\frac{P[\perm]}{P[\tilde \perm]} = \frac{\prod_{k=1}^j\sum_{\ell=k}^{\nitem}\exp(\val_{\tilde\perm_{\ell}}/\beta)}{\prod_{k=1}^j\sum_{\ell=k}^{\nitem}\exp(\val_{ \perm_{\ell}}/\beta)} = \frac{\prod_{k=2}^j\sum_{\ell=k}^{\nitem}\exp(\val_{\tilde \perm_{\ell}}/\beta)}{\prod_{k=2}^j\sum_{\ell=k}^{\nitem}\exp(\val_{ \perm_{\ell}}/\beta)}$$
 where the equality comes from the fact that the $k=1$ term is equal on the top and bottom. Noting that $\perm = \tilde \perm$ except in indices $1, j$, we can rewrite this as: 
 $$\frac{\prod_{k=2}^{j}\p{\exp(\val_{\tilde \perm_{j}}/\beta) + \sum_{\ell=k, \ell \ne j}^{\nitem}\exp(\val_{\tilde \perm_{\ell}}/\beta)}}{\prod_{k=2}^{j}\p{\exp(\val_{ \perm_{j}}/\beta) + \sum_{\ell=k, \ell \ne j}^{\nitem}\exp(\val_{ \perm_{\ell}}/\beta)}} = \frac{\prod_{k=2}^{j}\p{\exp(\val_{\perm_{1}}/\beta) + \sum_{\ell=k, \ell \ne j}^{\nitem}\exp(\val_{ \perm_{\ell}}/\beta)}}{\prod_{k=2}^{j}\p{\exp(\val_{ \perm_{j}}/\beta) + \sum_{\ell=k, \ell \ne j}^{\nitem}\exp(\val_{ \perm_{\ell}}/\beta)}}$$
 Because $\val_{\perm_1} > \val_{\perm_j}$ by assumption, this is satisfied. 

\textbf{Majorization:}
The property we wish to prove is the following: If the probability of a candidate of value $v_i$ being in the top $k$ for any $k\geq 1$ decreases, then the probability of any other candidate with value $v_j<v_i$ being in the top $k$ also decreases. 
Instead, we will prove a stronger version: for any set of items $S$ with $\abs{S} = k-2$ $v_i, v_k \not \in S$, if increasing accuracy decreases the probability of $v_i$ being ranked above all elements in $S \cup \{v_k\}$, then it also decreases the probability of $v_k$ being ranked above all elements in $S \cup \{v_i\}$. 

What's the probability that $v_i$ is higher than any of these other elements in $S$? 
$$\frac{\exp(v_i/\beta)}{S_{\beta} + \exp(v_i/\beta) + \exp(v_k/\beta)}$$
Suppose we consider some $\alpha < \beta$ (more accurate distribution). This shrinks the probability of $v_i$ being in the top $k$ if: 
$$\frac{\exp(v_i/\beta)}{S_{\beta} + \exp(v_i/\beta) + \exp(v_k/\beta)} > \frac{\exp(v_i/\alpha)}{S_{\alpha} + \exp(v_i/\alpha) + \exp(v_k/\alpha)}$$
Cross multiplying gives: 
$$\exp(\val_i\cd(1/\beta + 1/\alpha))  + \exp(\val_i/\beta) \cd (S_{\alpha} + \exp(\val_k/\alpha)) > \exp(\val_i \cd (1/\alpha + 1/\beta)) + \exp(\val_i/\alpha)\cd (S_{\beta} + \exp(\val_k/\beta))$$

$$ \exp(\val_i/\beta) \cd (S_{\alpha} + \exp(\val_k/\alpha)) > \exp(\val_i/\alpha)\cd (S_{\beta} + \exp(\val_k/\beta))$$
Strategically collecting terms gives: 
$$\exp\p{\val_i/\beta + \val_k/\alpha} + \exp(\val_i/\beta) \cd S_{\alpha} > \exp\p{\val_i/\alpha + \val_k/\beta} + \exp(\val_i/\alpha) \cd S_{\beta}$$
Because $\val_i > \val_k$ and $\alpha < \beta$, we know that $\exp\p{\val_i/\beta + \val_k/\alpha} > \exp\p{\val_i/\alpha + \val_k/\beta}$, and so if the above inequality holds, we must have that 
$$\exp(\val_i/\beta) \cd S_{\alpha}>\exp(\val_i/\alpha) \cd S_{\beta} \quad \Rightarrow \quad S_{\alpha}/S_{\beta} > \exp(\val_i \cd (1/\alpha - 1/\beta))$$ 

When is the same true for $v_k$? When: 
$$\frac{\exp(v_k/\beta)}{S_{\beta} + \exp(v_i/\beta) + \exp(v_k/\beta)} > \frac{\exp(v_k/\alpha)}{S_{\alpha} + \exp(v_i/\alpha) + \exp(v_k/\alpha)}$$
Or: 
$$\exp\p{\val_k/\beta + \val_i/\alpha} + \exp(\val_k/\beta) \cd S_{\alpha} > \exp\p{\val_k/\alpha + \val_i/\beta} + \exp(\val_k/\alpha) \cd S_{\beta}$$
By our previous analysis, we know that $\exp\p{\val_i/\beta + \val_k/\alpha} > \exp\p{\val_i/\alpha + \val_k/\beta}$, which satisfies one pair of inequalities. Because $\val_i > \val_k$, we know that 
$$S_{\alpha}/S_{\beta} > \exp(\val_i \cd (1/\alpha - 1/\beta)) \geq \exp(\val_k \cd (1/\alpha - 1/\beta))$$
which satisfies the other pair of terms. 
\end{proof}

\beyondwelfare*
\begin{proof}
First, we will consider the case that a candidate is using first-free. Here, a firm only picks a busy candidate if all in the window size (length $k$) are busy. The event where this occurs is: 
$$\sum_{\perm \in \mathcal{S}}P[\perm] \cd \prod_{i \in \perm_{1:k}} (1-p_i)$$
First, we consider the case where the window shrinks from $k$ to $k-1$: we wish to show that this \emph{increases} the chances that a busy candidate is picked. This event can be written as: 
$$\sum_{\perm \in \mathcal{S}}P[\perm] \cd \prod_{i \in \perm_{1:(k-1)}} (1-p_i)$$
Because $p_i \in (0, 1)$, we know that the product is always strictly larger: 
$$\prod_{i \in \perm_{1:(k-1)}} (1-p_i)\geq \prod_{i \in \perm_{1:k}} (1-p_i)$$

Next, we will consider the case where the window size is held constant, but accuracy increases (a shift from $\{P[\perm]\}$ to $\{P[\perm]'\}$). We wish to show that this increases the probability of the top $k$ candidates being busy (and thus, the probability of a busy candidate being picked). That is, we wish to show that: 
$$\sum_{\perm \in \mathcal{S}}P[\perm]' \cd \prod_{i \in \perm_{1:k}} (1-p_i) \geq \sum_{\perm \in \mathcal{S}}P[\perm] \cd \prod_{i \in \perm_{1:k}} (1-p_i)$$
Note that we can rewrite this as: 
$$\prod_{i \in [\nitem] \mid P[i \in \perm_{1:k}]'>0}(1-p_i) \cd P[i \in \perm_{1:k}]' \geq \prod_{i \in [\nitem] \mid P[i \in \perm_{1:k}]>0}(1-p_i) \cd P[i \in \perm_{1:k}] $$
that is: the probability that the candidate is in the top $k$, times the probability that it is busy. By Definition \ref{def:accuracybeyond} we know that if $P[i \in \perm_{1:k}]'<P[i \in \perm_{1:k}]$ for some $i$, then it is also lower for any $j>i$. Because $\val_i$ and probability of being free $\perm_i$ are inversely related, we know that the quantity on the lefthand side must have lower probability mass on the highest $p_i$ values (lowest $1-p_i$ values) and higher probability mass on lower $p_i$ values (higher $1-p_i$ values). This shows that the inequality must hold. 

Next, we turn to the setting where the firm is using first-busy as their strategy. Here, a firm selects a busy candidate \emph{unless} all of the top candidates are free. Thus, we will analyze the probability that all of the top candidates are free (and inverting the results for this quantity will tell us results about how often busy candidates are picked). 

The event where this occurs is: 
$$\sum_{\perm \in \mathcal{S}}P[\perm] \cd \prod_{i \in \perm_{1:k}} p_i$$
First, we consider the case where the window shrinks from $k$ to $k-1$: we wish to show that this \emph{increases} the chances that a free candidate is picked (thus, decreasing the probability that a busy candidate is picked). This event can be written as: 
$$\sum_{\perm \in \mathcal{S}}P[\perm] \cd \prod_{i \in \perm_{1:(k-1)}} p_i$$
Because $p_i \in (0, 1)$, we know that the product is always strictly larger: 
$$\prod_{i \in \perm_{1:(k-1)}} p_i\geq \prod_{i \in \perm_{1:k}} p_i$$
which means that the probability of all candidates in the window being free only \emph{increases}, as desired. \\

Next, we will consider the case where the window size is held constant, but accuracy increases (a shift from $\{P[\perm]\}$ to $\{P[\tilde \perm]\}$). We wish to show that this decreases the probability of the top $k$ candidates being free (and thus, increases probability of a busy candidate being picked). That is, we wish to show that: 
$$\sum_{\perm \in \mathcal{S}}P[\perm]' \cd \prod_{i \in \perm_{1:k}} p_i \leq \sum_{\perm \in \mathcal{S}}P[\perm] \cd \prod_{i \in \perm_{1:k}}p_i$$
Rewriting this gives: 
$$\prod_{i \in [\nitem] \mid P[i \in \perm_{1:k}]'>0}p_i \cd P[i \in \perm_{1:k}]' \leq \prod_{i \in [\nitem] \mid P[i \in \perm_{1:k}]>0}p_i \cd P[i \in \perm_{1:k}] $$

Again, we can prove this by appealing to the accuracy definition in Definition \ref{def:accuracybeyond}: we know that if the probability of candidates $i$ being in the top $k$ decreases, then it must decrease for all candidates $j>i$. Because the total probability of candidates being in the top $k$ must stay constant, this means that the probability of candidates with higher value (lower index) must increase. Because candidates with higher value have lower $p_i$, this implies that the total probability above must decrease with increased accuracy. 
\end{proof}

\beyondpicktop*
\begin{proof}
Suppose that a firm is using first-busy. Then, the firm picks the top ranked candidate if either a) the top candidate is busy or b) all top $k$ candidates are free. This is given by: 
$$\sum_{\perm \in S}P[\perm] \cd (1-p_{\perm_1}) + \sum_{\perm \in S}P[\perm] \cd \prod_{i \in \perm_{1:k}}p_i$$
First, we consider the case where the window size $k$ shrinks. Then this becomes: 
$$\sum_{\perm \in S}P[\perm] \cd (1-p_{\perm_1}) + \sum_{\perm \in S}P[\perm] \cd \prod_{i \in \perm_{1:k-1}}p_i$$
which is higher, so the probability of the top candidate being picked only increases. 
Next, we consider the case where the window stays the same but the accuracy increases. 
By Definition \ref{def:accuracybeyond} and our reasoning in the proof of Theorem \ref{thrm:beyondwelfare} we know that increasing the accuracy of the ranking tool only increases the total probability of the top $k$ elements being busy. Thus, this means that the first term increases and the second term increases - whether or not the total term increases or decreases depends on the relative change. 

Suppose that a firm is using first-free. Then, the firm picks the top ranked candidate if either a) the top candidate is free or b) all top $k$ candidates are busy. This is given by: 
$$\sum_{\perm \in S}P[\perm] \cd p_{\perm_1} + \sum_{\perm \in S}P[\perm] \cd \prod_{i \in \perm_{1:k}}(1-p_i)$$
First, we note that if the window size $k$ shrinks, then the overall term increases and the probability of picking the top candidate increases. 
Next, we consider the case where the window size stays constant but accuracy increases. Again by Definition \ref{def:accuracybeyond} and our reasoning in the proof of Theorem \ref{thrm:beyondwelfare}, the first term increases while the second term decreases: which dominates depends on the relative change in these two. 
\end{proof}

\beyondreduces*
\begin{proof}
Note that if we are in the superstar setting, there exists only a single pair of unique candidate types to compare, $\val_1>\val_2$, so all of the \enquote{votes} will go in the same direction. First, we will consider the setting where the first candidate is busy.

The quantity $G_{12}$ can be calculated as: 
$$G_{12} = \ratio \cd P[\perm] \cd (\val_2 - \val_1/\busypen_1) + P[\tilde \perm] \cd (\val_1 - \val_2/\busypen_2)$$
This is greater than 0 exactly when: 
$$\ratio \cd P[\perm] \cd (\val_2 - \val_1/\busypen_1) + P[\tilde \perm] \cd (\val_1 - \val_2/\busypen_2) > 0$$
$$ P[\tilde \perm] \cd (\val_1 - \val_2/\busypen_2) > \ratio \cd P[\perm] \cd (\val_1/\busypen_1 - \val_2)$$
By assumption, $\val_1/\busypen_1 - \val_2>0$ or else it would never be preferable to pick a busy candidate, which means this can be rewritten as: 
$$ \frac{\val_1 - \val_2/\busypen_2}{\ratio \cd (\val_1/\busypen_1 - \val_2)} > \frac{P[\perm]}{P[\tilde \perm] }$$
Note that the LHS is equal to $1/R$ as calculated in Algorithm \ref{algo:superstar}. For the $RHS$, $\perm$ is defined as a permutation with the high value candidate ranked $1$st, and $\tilde \perm$ is a permutation with it ranked in index $j$. Note that $P[\perm] \ne P[\perm^1]$, because the latter is the \emph{total probability} that the high value candidate is ranked 1st, while $P[\perm]$ is the probability of a single permutation occurring where the high value candidate is ranked first. However, $P[\perm^1] = \nitem \cd P[\perm]$ and $P[\perm^j] = \nitem \cd P[\tilde \perm]$, and so $P[\perm]/P[\tilde \perm] = P[\perm^1]/P[\perm^j]$, as desired. Similar analysis holds for the setting where the first candidate is free. 
\end{proof}
\beyondprop*
\begin{proof}
We can rewrite these conditions as four properties: 
\begin{enumerate}
    \item Algorithm \ref{algo:approximatebeyondsuperstar} either preferentially picks free or busy candidates (e.g. if it ever picks a free candidate in index $j>1$, then it \emph{never} picks a busy candidate in any $k>1$, and vice versa). 
    \item If Algorithm \ref{algo:approximatebeyondsuperstar} would pick a candidate in index $j$, then it would also pick a candidate in index $k<j$ (assuming corresponding free/busy status). 
    \item Increased accuracy decreases window size. 
    \item If a firm is preferentially picking busy candidates, increased $\busypen$ shrinks the window size (and similarly increases window size if the firm is preferentially picking free candidates). 
\end{enumerate}
Note that properties 1 and 2 together tell us that this strategy reduces down to fixing a window size and preferentially picking either free or busy candidates within that window size.

We will prove these properties in sequence. First, we note that Lemma \ref{lem:gumbelsataccuracy} gives us a closed-form solution for the relative probability of finding a permutation with two items flipped: 
$$\frac{P[\perm]}{P[\tilde \perm]} = \frac{\prod_{k=2}^{j}\p{\exp(\val_{\perm_{1}}/\beta) + \sum_{\ell=k, \ell \ne j}^{\nitem}\exp(\val_{ \perm_{\ell}}/\beta)}}{\prod_{k=2}^{j}\p{\exp(\val_{ \perm_{j}}/\beta) + \sum_{\ell=k, \ell \ne j}^{\nitem}\exp(\val_{ \perm_{\ell}}/\beta)}} $$
We can note that, given fixed $\val_{\perm_1}, \val_{\perm_j}$, this is fixed when all other terms within the sum are minimized. In particular, this means that we must have $\perm$ be the ordering with every item besides $i, k$ correctly ordered.

\textbf{Property 1: Algorithm \ref{algo:approximatebeyondsuperstar} either preferentially picks free or busy candidates:} We will show that if the strategy preferentially picks a candidate who is busy (resp. free) in index $j>1$, then it will \emph{never} pick a candidate who is free (resp. busy) in any index $k>1$. WLOG, we will consider a strategy that preferentially picks a free candidate in index $j$. If this occurs, then $\abs{G_j} > \abs{G_1}$. For a particular pair of candidates $i, k$, we have $G_{ik}$ given by:

    $$G_{ik} = \ratio_{i, k} \cd P[\perm] \cd (\val_{k} - \val_{i}/\busypen) + P[\tilde \perm] \cd (\val_{i} - \val_{k}/\busypen)$$

    Next, we will consider the hypothetical  $\tilde G_{i, k}$ if a) the first candidate were \emph{free} and b) the choice is picking the candidate ranked 1st and between any possible index $\ell$. The relevant quantity for candidates pairs $i, k$:  
    $$\tilde G_{ik} =  P[\perm'] \cd (\val_{k}/\busypen - \val_{i}) + \ratio_{i, k} \cd P[\tilde \perm'] \cd (\val_{i}/\busypen - \val_{k})$$
    We will show that if $G_{ik} > 0$, we must have $\tilde G_{ik} < 0$. By assumption, $P[\perm]> P[\tilde \perm]$ and $P[\perm']> P[\tilde \perm']$ and $\val_i > \val_k$.  Note that $\tilde G_{ik}$ has two components, $C_1 = \val_k - \val_i/\busypen$, and $C_2 = \val_i - \val_k/\busypen$, and that they are inverted in $\tilde G_{ik}$. Analyzing the terms case by case: 
    \begin{itemize}
        \item If $C_1, C_2 > 0$, then $\tilde G_{ik}$ can be written as $ P[\perm'] \cd (-C_1) + \ratio_{ik} \cd P[\tilde \perm'] \cd (-C_2)$, and thus is negative. 
        \item We cannot have $C_1, C_2 < 0$ because we know $G_{ik} >0$. 
        \item If $C_1 < 0, C_2 > 0$, then we know because $G_{ik}>0$ 
        $$0 < G_{ik} \leq \ratio_{ik} \cd C_1 + C_2 $$
        we have $\abs{\ratio_{ik} \cd C_1} < C_2$, which implies that 
        $$\tilde G_{ik} \leq -C_2 - \ratio \cd C_1 \leq 0$$
        \item Note that we cannot have $C_1 >0, C_2 <0$, as we must have $C_1 \leq C_2$ because: 
    $$C_1 \leq C_2$$
    $$v_k - v_i /\busypen \leq v_i - v_k/\busypen$$
    $$\val_k \cd (1 + 1/\busypen) \leq v_i \cd (1 + 1/\busypen)$$
    which is satisfied by assumption. 
    \end{itemize}

\textbf{Property 2: if Algorithm \ref{algo:approximatebeyondsuperstar} would pick a candidate in index $j$, then it would also pick a candidate in index $k<j$} In order to prove this, we will show that if Algorithm \ref{algo:approximatebeyondsuperstar} would pick a candidate in index $j$, then it would also pick a candidate in index $k<j$ in the same setting (e.g. assuming the first candidate and $j$th/$k$th candidate have corresponding statuses). 

Note that in Algorithm \ref{algo:approximatebeyondsuperstar}, the only portions that rely on the index of the candidate is the $P[\perm], P[\tilde \perm]$ calculation. From our analysis at the beginning of this proof, we note that: 
$$\frac{P[\perm_{i, k, j}]}{P[\tilde \perm_{i, k, j}]} = \frac{P[\perm_{i, k, j-1}]}{P[\tilde \perm_{i, k, j-1}]} \cd \frac{\exp(\val_{\perm_{1}}/\beta) + \sum_{\ell>j}^{\nitem}\exp(\val_{ \perm_{\ell}}/\beta)}{\exp(\val_{ \perm_{j}}/\beta) + \sum_{\ell>j}^{\nitem}\exp(\val_{ \perm_{\ell}}/\beta)}$$
Because $\val_{\perm_1} > \val_{\perm_j}$ by assumption, the ratio on the RHS is greater than 1. Note that a larger value of $P[\perm]$ as compared to $P[\tilde \perm]$ only makes it more challenging to pick the candidate ranked $j$th (lower threshold to pick the candidate ranked 1st).

\textbf{Property 3: Increased accuracy shrinks window size: } Recall that the only component of Algorithm \ref{algo:approximatebeyondsuperstar} that relies on window size is the ratio $P[\perm]/P[\tilde \perm]$. Increased accuracy (decreased $\beta$) increases this ratio because it increases $\exp(\val_{\perm_1}/\beta)$ more quickly than $\exp(\val_{\perm_j}/\beta)$.

\textbf{Property 4: Increasing $\busypen$ shrinks the window size if a firm is preferentially picking busy candidates (and increases it otherwise):} Note that the only place that Algorithm \ref{algo:approximatebeyondsuperstar} relies on $\busypen$ is in calculating $G_{ik}$. A firm is preferentially picking busy candidates if it ever picks a busy candidate ranked in index $j$. Increasing $\busypen$ decreases $\val_k/\busypen, \val_i/\busypen$ and thus decreases $G_{ik}$ for when the first candidate is free, either keeping constant or shrinking the window size. Similarly, if the first candidate is busy, increasing $\busypen$ shrinks $\val_i/\busypen, \val_k/\busypen$ and increases $G_{ik}$, either keeping constant or increasing the window size. 
\end{proof}

\beyondstratopterr*
\begin{proof}

\textbf{Exact conditions:} We will find it helpful to begin by exploring the exact conditions for the optimal decision. First, when the \textbf{first candidate is free} it is optimal to pick the $j$th (busy) candidate exactly whenever: 
$$\mathbb{E}_{\perm \sim \permdist}[\val_{\perm_1} \mid \avail] \leq  \mathbb{E}_{\perm \sim \permdist}[\val_{\perm_j} \mid \avail]/\busypen $$
Applying Bayes rule tells us that we can rewrite this as: 
$$\sum_{\perm \in \permdist} \frac{P[\avail \mid \perm ] \cd P[\perm]}{P[\avail]} \cd \val_{\perm_1}\leq  \sum_{\perm \in \permdist} \frac{P[\avail \mid \perm ] \cd P[\perm]}{P[\avail]} \cd \val_{\perm_j} /\busypen$$
Pulling over: 
$$0 \leq \frac{1}{P[\avail]} \sum_{\perm \in \permdist} P[\avail \mid \perm] \cd P[\perm] \cd (\val_{\perm_j}/\busypen - \val_{\perm_1})$$
Next, we rewrite strategically by grouping together permutations $\perm$ that have elements in index 1, $j$ ordered correctly, and corresponding permutations $\tilde \perm$ that are identical to $\perm$, but with values $1, j$ exactly inverted: 

$$0 \leq \frac{1}{P[\avail]} \sum_{\perm \in \permdist, \val_{\perm_1}> \val_{\perm_j}} P[\avail \mid \perm] \cd P[\perm] \cd (\val_{\perm_j}/\busypen - \val_{\perm_1}) + P[\avail \mid \tilde \perm] \cd P[\tilde \perm] \cd (\val_{\tilde \perm_j}/\busypen- \val_{\tilde \perm_1}) $$
which we can rewrite as: 
\begin{equation}\label{eq:exactlyfirstfree}
    0 \leq \frac{1}{P[\avail]} \sum_{\perm \in \permdist, \val_{\perm_1}> \val_{\perm_j}} P[\avail \mid \perm] \cd \p{P[\perm] \cd (\val_{\perm_j}/\busypen_{\perm_j} - \val_{\perm_1}) +\ratio_{\perm_1, \perm_j} \cd  P[\tilde \perm] \cd (\val_{\perm_1}/\busypen_{\perm_1} - \val_{\perm_j})} 
\end{equation}
Collecting pairs $i, k$ gives us: 

\begin{equation}\label{eq:exactlyfirstfree}
    0 \leq \frac{1}{P[\avail]} \sum_{i, k \mid \val_i > \val_k} \p{\p{\sum_{\perm \mid \perm_1 = i, \perm_j = k} P[\avail \mid \perm] \cd P[\perm]} \cd \p{(\val_{\perm_j}/\busypen_{\perm_j} - \val_{\perm_1})} + \p{\sum_{\tilde \perm \mid \tilde \perm_1 = k, \tilde \perm_j = i} P[\avail \mid \perm] \cd P[\tilde \perm] }\cd \ratio_{\perm_1, \perm_j} \cd (\val_{\perm_1}/\busypen_{\perm_1} - \val_{\perm_j})} 
\end{equation}

For any particular pair of elements $v_i > v_k$, the relevant gap in Algorithm \ref{algo:approximatebeyondsuperstar} is given by: 
$$G_{ik} = P[\perm'] \cd (\val_{k}/\busypen_{k} - \val_{i}) + \ratio_{i, k} \cd P[\tilde \perm'] \cd (\val_{i}/\busypen_{i} - \val_{k})$$
where $\perm', \tilde \perm'$ are defined specifically where $\perm$ is the permutation with every item ordered correctly, except for with item $i$ ranked 1st and item $k$ ranked in index $j$, and $\tilde \perm$ is identical to $\perm$ except with items $i, k$ swapped. By prior analysis in Lemma \ref{lem:beyondprop}, we know that $P[\perm]/P[\tilde \perm]$ is maximized with $\perm', \tilde \perm'$ as defined. Define $\mathcal{P}_{i, k, j}$ as the set of all pairs of permutations $\perm, \tilde \perm$ such that $\perm = \tilde \perm$ except with items $i, k$ in indices $1, j$ swapped. 

Thus, if $G_{ik} > 0$, then we know that the corresponding sum for every $\perm, \tilde \perm$ in $\mathcal{P}_{i, k, j}$ is also positive. Thus, we can analyze the error as follows: \\
If Algorithm \ref{algo:approximatebeyondsuperstar} picks the candidate in index $j$, then there is 0 error from pairs $i, k$ with $G_{ik} >0$, because by our prior reasoning the term within the sum is positive for all permutations in $\mathcal{P}_{i, k, j}$. However, there is error from pairs $i, k$ with $G_{ik} <0$ (pairs that \enquote{voted} to pick the candidate ranked 1st). Using the $\perm', \tilde \perm'$ as defined in Algorithm \ref{algo:approximatebeyondsuperstar} are the pair of permutations that makes it \emph{most} easy to \enquote{vote} for picking the candidate ranked 1st, and therefore we can upper bound the error associated with picking the candidate ranked $j$th from this pair $i, k$ as: 
$$\frac{P[\avail \mid \perm]}{P[\avail]}\sum_{\perm, \tilde \perm \in \mathcal{P}_{i, k, j}}-\p{P[\perm] \cd (\val_{\perm_j}/\busypen_{\perm_j} - \val_{\perm_1}) +\ratio_{\perm_1, \perm_j} \cd P[\tilde \perm] \cd (\val_{\perm_1}/\busypen_{\perm_1} - \val_{\perm_j})}  \leq \max{\perm, \rho} \frac{P[\avail \mid \perm]}{P[\avail \mid \rho]}\cd  P[\mathcal{P}_{i, k, j}] \cd (\val_i - \val_k/\busypen)$$ 
For conciseness, in the analysis below we will omit the $\max{\perm, \rho} \frac{P[\avail \mid \perm]}{P[\avail \mid \rho]}$ term, which is a coefficient on all bounds below. \\
Next, if Algorithm \ref{algo:approximatebeyondsuperstar} picks the candidate ranked 1st, then for pairs with $G_{ik}>0$ the error is upper bounded by: 
$$P[\mathcal{P}_{i, k, j}] \cd \ratio_{1, k} \cd(\val_i/\busypen - \val_k)$$
For pairs with $G_{ik} < 0$ the error comes from the fact that $G_{ik}$ is biased towards picking the top element by its use of $\perm', \tilde \perm '$. The error is upper bounded by the gap between this and the most conservative permutations, $\perm'', \tilde \perm ''$ where $\perm ''$ is created by having every item inverted. Then, the error is bounded by: 
$$P[\mathcal{P}_{i, k, j}] \cd \max(-G_{ik}', 0)$$

By similar analysis, when the \textbf{first candidate is busy} it is optimal to pick the $j$th (free) candidate exactly whenever: 
$$\mathbb{E}_{\perm \sim \permdist}[\val_{\perm_1} \mid \avail]/\busypen \leq  \mathbb{E}_{\perm \sim \permdist}[\val_{\perm_j} \mid \avail] $$
Applying Bayes rule tells us that we can rewrite this as: 
$$\sum_{\perm \in \permdist} \frac{P[\avail \mid \perm ] \cd P[\perm]}{P[\avail]} \cd \val_{\perm_1}/\busypen\leq  \sum_{\perm \in \permdist} \frac{P[\avail \mid \perm ] \cd P[\perm]}{P[\avail]} \cd \val_{\perm_j}$$
Pulling over: 
$$0 \leq \frac{1}{P[\avail]} \sum_{\perm \in \permdist} P[\avail \mid \perm] \cd P[\perm] \cd (\val_{\perm_j} - \val_{\perm_1}/\busypen)$$
Next, we rewrite strategically by grouping together permutations $\perm$ that have elements in index 1, $j$ ordered correctly, and corresponding permutations $\tilde \perm$ that are identical to $\perm$, but with values $1, j$ exactly inverted: 

$$0 \leq \frac{1}{P[\avail]} \sum_{\perm \in \permdist, \val_{\perm_1}> \val_{\perm_j}} P[\avail \mid \perm] \cd P[\perm] \cd (\val_{\perm_j} - \val_{\perm_1}/\busypen) + P[\avail \mid \tilde \perm] \cd P[\tilde \perm] \cd (\val_{\tilde \perm_j}- \val_{\tilde \perm_1}/\busypen) $$
which we can rewrite as: 
\begin{equation}\label{eq:exactlyfirstbusy}
    0 \leq \frac{1}{P[\avail]} \sum_{\perm \in \permdist, \val_{\perm_1}> \val_{\perm_j}} P[\avail \mid \perm] \cd \p{ P[\perm] \cd (\val_{\perm_j} - \val_{\perm_1}/\busypen_{\perm_1}) + P[\tilde \perm] \cd \ratio_{\perm_1, \perm_j} \cd (\val_{\perm_1} - \val_{\perm_j}/\busypen_{\perm_j})} 
\end{equation}
For any particular pair of elements $v_i > v_k$, the relevant gap in Algorithm \ref{algo:approximatebeyondsuperstar} is given by: 
$$G_{ik} = P[\perm] \cd (\val_{k} - \val_{i}/\busypen_{i}) + P[\tilde \perm]\cd \ratio_{i, k} \cd (\val_{i} - \val_{k}/\busypen_{k})$$
By using similar analysis to above, we know that when Algorithm \ref{algo:approximatebeyondsuperstar} picks the candidate in index $j$, for pairs in $G_{ik}>0$ there is 0 error, and for pairs with $G_{ik} < 0$ there is error upper bounded by: 
$$\max{\perm, \rho} \frac{P[\avail \mid \perm]}{P[\avail \mid \rho]} \cd P[\mathcal{P}_{i, k, j}]\cd (\val_i/\busypen - \val_k)$$
Similarly, when Algorithm \ref{algo:approximatebeyondsuperstar} picks the candidate ranked 1st, for pairs in $G_{ik}>0$ there is regret upper bounded by: 
$$\max{\perm, \rho} \frac{P[\avail \mid \perm]}{P[\avail \mid \rho]} \cd P[\mathcal{P}_{i, k, j}]\cd (\val_i - \val_k/\busypen)$$
and for pairs in $G_{ik}<0$ there is regret upper bounded by:
$$P[\mathcal{P}_{i, k, j}] \cd \max(- G_{ik}', 0)$$
\end{proof}

\end{document}